\documentclass[12pt,reqno]{article}
\usepackage{amsmath,amsthm,commath,mathrsfs,amssymb,extarrows}
\usepackage{dsfont}
\usepackage{relsize}

\usepackage{pgfplots}

\usepackage{tikz}
\usetikzlibrary{arrows,intersections}

\usepackage[colorlinks = true,urlcolor = blue,citecolor = blue,breaklinks]{hyperref}

\usepackage{amsfonts}
\usepackage{graphicx}
\usepackage{enumerate}
\usepackage{subcaption}
\usepackage[right,pagewise,displaymath, mathlines]{lineno}
\usepackage{epstopdf}
\usepackage{color}
\usepackage{multirow}
\usepackage{authblk}
\usepackage[round]{natbib}
\usepackage{float}
\restylefloat{figure,table}

\usepackage{booktabs}

\usepackage{colortbl}

\usepackage{geometry}
\newtheorem{theorem}{Theorem}[section]

\newtheorem{proposition}{Proposition}[section]

\theoremstyle{definition}

\newtheorem{remark}{Remark}[section]

\definecolor{darkread}{rgb}{0.7, 0, 0}
\definecolor{darkbrown}{rgb}{0.55, 0.2, 0.15}
\definecolor{darkblue}{rgb}{0.1,0.1,0.6}
\definecolor{darkgreen}{rgb}{0.1,0.5,0.2}

\title{\Large\bf Intra-Class Correlation Coefficient Ignorable Clustered Randomized Trials for Detecting Treatment Effect Heterogeneity}

\author[1,2,3]{Chen Yang\thanks{Corresponding author; e-mail \href{mailto:chen.yang@mountsinai.org}{chen.yang@mountsinai.org}}}

\author[1,2,3]{M\'arcio A. Diniz}

\author[4]{Deukwoo Kwon}

\author[1,2,3]{Madhu Mazumdar}

\affil[1]{\normalsize Department of Population Health Science and Policy, Icahn School of Medicine at Mount Sinai, New York, USA}

\affil[2]{\normalsize Institute for Healthcare Delivery Science, Mount Sinai Health System, New York, USA}

\affil[3]{\normalsize Tisch Cancer Institute, Icahn School of Medicine at Mount Sinai, New York, USA}

\affil[4]{\normalsize Division of Clinical and Translational Sciences, Department of Internal Medicine, University of Texas Health Science Center at Houston, Texas, USA}

\date{}

\usepackage{setspace}
\usepackage[bottom]{footmisc}

\begin{document}

\maketitle
\vspace{-10mm}

\noindent
{\bf Abstract.}
Accurately estimating the intra-class correlation coefficient (ICC) is crucial for adequately powering clustered randomized trials (CRTs). Challenges arise due to limited prior data on the specific outcome within the target population, making accurate ICC estimation difficult. Furthermore, ICC can vary significantly across studies, even for the same outcome, influenced by factors like study design, participant characteristics, and the specific intervention. Power calculations are extremely sensitive to ICC assumptions. Minor variation in the assumed ICC can lead to large differences in the number of clusters needed, potentially impacting trial feasibility and cost. 

\bigskip
This paper identifies a special class of CRTs aiming to detect the treatment effect heterogeneity, wherein the ICC can be completely disregarded in calculation of power and sample size. This result offers a solution for research projects lacking preliminary estimates of the ICC or facing challenges in their estimate. Moreover, this design facilitates power improvement through increasing the cluster sizes rather than the number of clusters, making it particular advantageous in the situations where expanding the number of clusters is difficult or costly.

\bigskip
This paper provides a rigorous theoretical foundation for this class of ICC-ignorable CRTs, including mathematical proofs and practical guidance for implementation. We also present illustrative examples to demonstrate the practical implications of this approach in various research contexts in healthcare delivery.
\smallskip



\section{Introduction}\label{sec1: intro}

Interest in investigating heterogeneity of treatment effect (HTE) \citep{yang2020sample} in cluster randomized trials (CRTs) is growing, particularly for interventions implemented at the cluster level (e.g.,  evaluation of clinical decision support systems or training programs for physicians caring for multiple patients in clinics as clusters) \citep{yang2024power}. Although randomization occurs at the cluster level in CRTs, data collection and subsequent analyses are conducted at the individual patient level. Individual outcomes are correlated within clusters, an effect quantified by the intra-class correlation coefficient (ICC). Additionally, individuals within specific subgroups (e.g., sex, age or race/ethnicity) may respond similarly to the intervention, but these responses can differ significantly across subgroups. Mixed-effects models or generalized estimating equations are commonly used to analyze CRT data. These models typically include cluster and treatment as random effect, the subgroup and treatment interaction as fixed effects to estimate HTE.

The overall treatment effect (OTE) often serves as the primary endpoint in trials, while subgroup analyses examining HTE are often secondary \citep{nowalk2014increasing,sequist2010cultural,garbutt2015cluster,kruis2014effectiveness}. Trials primarily focused on detecting HTE often lack sufficient power, as sample size calculations are typically based on the OTE. This issue was particularly pronounced in health disparity research, stimulating substantial research in this area.  \citet{yang2024power} formulated a power calculation method for detecting HTE in parallel CRTs with equal cluster sizes. \citet{tong2022accounting} extended this framework to include scenarios with unequal cluster sizes. For CRTs with binary outcomes, \citet{maleyeff2023sample} derived the design effect under equal cluster sizes using a logistic mixed-effect model. \citet{tong2023accounting} provided the design effect for CRTs with equal cluster sizes when outcomes are missing completely at random. \citet{li2023designing} further explored the design effect in three-level CRTs, considering different levels of randomization. These contributions collectively provide a comprehensive framework for sample size determination and power calculation in CRTs aimed at detecting HTE as a function of the ICC. Note that the aforementioned references \citep{yang2024power,tong2022accounting} implicitly rely on the assumption of homogeneous subgroup distributions within each cluster, which might require additional verification in practical applications.

Accurately estimating the ICC is crucial for adequately powering CRTs. Challenges arise due to limited prior data on the specific outcome within the target population, making accurate ICC estimation difficult. Furthermore, ICC can vary significantly across studies, even for the same outcome, influenced by factors like study design, participant characteristics, and the specific intervention. Power calculations are extremely sensitive to ICC assumptions. Minor variation in the assumed ICC can lead to large differences in the number of clusters needed, potentially impacting trial feasibility and cost. Underestimating the ICC increases the risk of underpowering the trial, leading to higher probability of failing to detect a true treatment effect. Addressing this challenge requires incorporating ICC uncertainty into power calculations. While this approach can be complex and may require specialized statistical software, it is crucial for robust trial design. Sensitivity analyses are commonly employed to assess the impact of ICC uncertainty on sample size requirements. However, these analyses can sometimes yield a wide range of sample size recommendations, potentially limiting their practical utility.

To address this problem, in this paper, we identify a class of parallel CRTs specifically designed to detect HTE, where the ICC can be completely disregarded in sample size determination and power calculation (hereafter referred to as ``ICC-ignorable CRTs’’). Furthermore, the sample size and power calculation formulas for ICC-ignorable CRTs exhibit patterns closely resembling those of individual-level randomized trials. This allows for enhanced statistical power in CRTs by increasing either the number of clusters or the cluster sizes, providing greater flexibility in study design and potentially increasing the feasibility of conducting adequately powered trials. This is particularly advantageous for underpowered CRTs where increasing the number of clusters is infeasible. Moreover, our ICC-ignorable design requires only that the proportions of subgroups be identical within each cluster, which is more accessible than the assumption of homogeneous subgroup distribution and readily verifiable compared with the potential extra steps to verifying the homogeneous distribution within each cluster in practice.

We also verify that the ICC-ignorable property applies to CRTs with unequal cluster sizes. In such cases, the pre-randomization expected power (PREP), defined as the average power across all possible treatment assignments, is influenced by randomization schemes \citep{ouyang2021crtpowerdist}. Therefore, we focus on the power and sample size calculation for ICC-ignorable CRTs under the random allocation rule \citep{lachin1988properties}, which ensures a 1:1 ratio between intervention and control arms. We propose an approximation method for power and sample size calculation in ICC-ignorable CRTs and evaluate its performance using Monte Carlo simulation.

The rest of the paper is organized as follows: Section~\ref{sec2: model} introduces the statistical model and parameters for general parallel CRTs. Section~\ref{sec3: main} presents our main findings regarding the ICC-ignorable property for both equal and unequal cluster sizes. This section also reviews the random allocation rule and provides an approximation of the variance for the generalized least squares estimation of the HTE under this rule. Section~\ref{sec4: simulation} details the operational characteristics of our simulations, which verify the ICC- ignorable property and evaluate the performance of our proposed variance approximation under the random allocation rule. Section~\ref{sec5: application} illustrates the sample size and power calculations for detecting HTE using two existing CRTs. Finally, Section~\ref{sec6: discussion} concludes the paper with discussions and future directions.

\section{Statistical Model}\label{sec2: model}

Assume there are $I$ ($I\ge 2$) clusters included in the parallel CRT with two arms: intervention and control. The cluster sizes are denoted by a vector $\boldsymbol{m} = \begin{bmatrix} m_1 & \cdots & m_I \end{bmatrix}^{\top}$ with element $m_i$ representing the number of participants within cluster $i$ for $i=1,\ldots,I$. The  parallel CRT design is characterized by a random vector $\boldsymbol{W} = \begin{bmatrix} W_1 & \cdots & W_I \end{bmatrix}^{\top}$ with element $W_i=1$ representing that cluster $i$ is assigned to the intervention arm and $W_i = 0$ otherwise. Within each cluster $i$, participants are characterized by a categorical variable with $(p+1)$ levels. From these levels, one is designated as the reference level, henceforth referred to as the reference subgroup. The categorical variable can be transferred into a binary matrix $\mathbf{X}_i$ with dimension $(m_i\times p)$ and compresses $p$ dummy variables with element $X_{ijl} = 1$ and $X_{ijl} = 0$ 
representing that participant $j\in\{1,\ldots,m_i\}$ in cluster $i$ belongs to the subgroup $l\in\{1,\ldots,p\}$ and $\sum^p_{l=1}X_{ijl} = 0$ representing that participant $j$ from cluster $i$ belongs to the reference subgroup, respectively. Moreover, the continuous outcome $\boldsymbol{Y}_i\in\mathbb{R}^{m_i}$ is assumed to be generated by the following linear mixed-effect model (LMM):
\begin{equation}\label{model}
\boldsymbol{Y}_i = (\beta_1 + \beta_2W_i)\mathbf{1}_{m_i} + \mathbf{X}_i(\boldsymbol{\beta}_3 + \boldsymbol{\beta}_4W_i) + (\gamma_i\mathbf{1}_{m_i} + \boldsymbol{\epsilon}_i),
\end{equation}
for $i=1,\ldots,I$, where $\mathbf{1}_{m_i}$ is the $m_i$-dimensional vector with all elements equal to $1$, $\gamma_i\sim\mathcal{N}(0, \sigma^2_{\gamma})$ is the random intercept at the cluster level and $\boldsymbol{\epsilon}_i\sim\mathcal{N}_{m_i}(\mathbf{0}, \sigma^2_{\epsilon}\mathbf{I}_{m_i})$ represents the individual level random error ($\mathbf{I}_{m_i}\in\mathbb{R}^{m_i\times m_i}$ is the identity matrix of dimension $m_i\times m_i$), $\beta_1$ is the grand intercept for the reference subgroup in the control arm, $\beta_2$ is the treatment effect for the reference subgroup (those participant $j$ with $\sum^p_{l=1}X_{ijl} = 0$ within cluster $i$), $\boldsymbol{\beta}_3\in\mathbb{R}^p$ represents the subgroup fixed effect, and $\boldsymbol{\beta}_4\in\mathbb{R}^p$ represents the HTE. We further assume that $\gamma_i$ is independent of $\boldsymbol{\epsilon}_j$ for all combinations of cluster pair $\{(i,k)\}_{1\le i,k\le I}$.

The power of the test with null hypothesis $H_0:\boldsymbol{\beta}_4 = \mathbf{0}$ is usually calculated by a Wald test which requires the unconditional variance of the generalized least square (GLS) estimator $\hat{\boldsymbol{\beta}}_4$ of the HTE, $\boldsymbol{\beta}_4$. For the parallel CRT with equal cluster sizes, i.e. $\boldsymbol{m} = m\mathbf{1}_I$ for some $m\in\mathbb{N}_+$, Yang et. al. (2020) provides the following result:
$$\mathbb{V}ar(\hat{\beta}_4) = \frac{\sigma^2_4(m, \rho, \rho_x, \overline{W}, \sigma^2_{y|x})}{I} + o\left(\frac{1}{I}\right)$$
for the case that $\boldsymbol{\beta}_4$ is 1-dimensional, where
\begin{equation}\label{existing result}
\sigma^2_4(m, \rho, \rho_x, \overline{W}, \sigma^2_{y|x}) = \frac{\sigma^2_{y|x}(1 - \rho)\{1 + (m - 1)\rho\}}{m\sigma^2_x\overline{W}(1 - \overline{W})\{1 + (m - 2)\rho - (m - 1)\rho_x\rho\}},
\end{equation}
$\overline{W}=\mathbb{E}[W_i]$ for all $i\in 1,\ldots,I$, $\sigma^2_x$ is the variance of all $X_{ij1}$'s for all $i=1,\ldots,n$ and $j=1,\ldots,m$, $\rho$ is the ICC of the outcome and $\rho_x$ is the common within-cluster correlation coefficient of $X_{ij1}$'s \citep{yang2020sample}. Beside, the authors also provide an extension of \eqref{existing result} for multi-dimensional $\boldsymbol{\beta}_4$. More generally, \citet{tong2022accounting} provides the following result:
$$\mathbb{V}ar(\hat{\boldsymbol{\beta}}_4) = \frac{\sigma^2_4(F, \bar{m}, \rho, \rho_x, \overline{W}, \sigma^2_{y|x})}{I} + o\left(\frac{1}{I}\right)$$
for 1-dimensional $\boldsymbol{\beta}_4$, where
\begin{equation}\label{existing result general}
\sigma^2_4(F, \bar{m}, \rho, \rho_x, \overline{W}, \sigma^2_{y|x}) = \frac{\sigma^2_{y|x}(1 - \rho)}{\sigma^2_x\overline{W}(1 - \overline{W})[\bar{m} + (1 - \rho_x)\bar{p} + \rho_x\bar{q}]},
\end{equation}
$\bar{m} = \mathbb{E}[m_i]$,
$$\bar{p} = \mathbb{E}\left[\frac{-m_ip}{1 + (m_i - 1)\rho}\right],\quad \bar{q} = \mathbb{E}\left[\frac{-m^2_i\rho}{1 + (m_i - 1)\rho}\right].$$
by assuming cluster sizes $m_1,\ldots,m_I$ are drawn randomly from some common distribution function $F$. Futhermore, a multi-dimensional extension of \eqref{existing result general} can also be found in \citet{tong2022accounting}. In this paper, we shall identify a specific class of parallel CRTs such that $\sigma^2_4$ does not depend on $\rho$ (or $\sigma^2_{\gamma}$). 

\section{The ICC-ignorable parallel CRTs}\label{sec3: main}

Intuitively, if all clusters exhibit similar patterns, then it is equivalent to having small variability of cluster random effect $\gamma_i$, meaning small value for the variance $\sigma^2_{\gamma}$. When the ICC $\rho$ should also be very small. If the ICC $\rho$ approaches zero the CRT effectively becomes an individual-level randomized trial. In this section, we examine parallel CRTs for testing $H_0:\boldsymbol{\beta}_4 = 0$ by fixing the group proportions within each cluster. We demonstrate that the ICC $\rho$ can be disregarded in power calculations for such designs.

\subsection{Constant cluster size scenario}\label{sec: main equal}

First let us consider the particular case with $\boldsymbol{m} = m\mathbf{1}_I$ for some cluster size $m\ge 2$, for which \eqref{existing result} holds for 1-dimensional $\boldsymbol{\beta}_4$ as we have only two subgroups. We further assume that
\begin{equation}\label{fixed prevalence}
\mathbf{1}^{\top}_m\boldsymbol{X}_i = k = \theta m,\quad \theta\in(0, 1)
\end{equation}
for all $i=1,\ldots,I$, where $\boldsymbol{X}_i$ is a vector of $X_{ij1}$'s. Notice that by imposing condition \eqref{fixed prevalence} the probability for individual $j$ in cluster $i$ to have attribute $X_{ij1} = 1$ is equivalent to the selection probability of a simple random sample of size $k$ from a population of size $m$, i.e.
$$\mathbb{P}(X_{ij1} = 1) = \frac{\displaystyle{m - 1\choose k - 1}}{\displaystyle{m\choose k}} = \frac{k}{m} = \theta.$$
Similarly, the joint probability of $X_{ij1}$ and $X_{il1}$ is given by
$$\mathbb{P}(X_{ij1} = 1, X_{il1} = 1) = \frac{\displaystyle{m - 2\choose k - 2}}{\displaystyle{m\choose k}} = \frac{k(k - 1)}{m(m - 1)} = \frac{(k - 1)\theta}{m - 1},\quad l\neq j.$$
Hence
$$\sigma^2_x = \mathbb{V}av(X_{ij1}) = \theta(1 - \theta),\quad \mathbb{C}ov(X_{ij1}, X_{il1}) = \theta\left(\frac{k - 1}{m - 1} - \frac{k}{m}\right) = -\frac{\theta(1 - \theta)}{(m - 1)},$$
which leads to
\begin{equation}\label{rho x}
\rho_x = -\frac{1}{m - 1}.
\end{equation}
Plugging \eqref{rho x} in \eqref{existing result} yields
\begin{equation}\label{result particular}
\sigma^2_4(m, \rho, \rho_x, \overline{W}, \sigma^2_{y|x}) = \frac{\sigma^2_{y|x}(1 - \rho)}{m\sigma^2_x\overline{W}(1 - \overline{W})} = \frac{\sigma^2_{\epsilon}}{m\theta(1 - \theta)\overline{W}(1 - \overline{W})}
\end{equation}
as the ICC $\rho = \sigma^2_{\gamma}/\sigma^2_{y|x}$ and $\sigma^2_{y|x} = \sigma^2_{\gamma} + \sigma^2_{\epsilon}$. Notice that neither $\rho$ nor $\sigma^2_{\gamma}$ appears in \eqref{result particular}, which means by fixing the within-cluster proportion of each group, the random effect $\gamma_i$ is completely ignored in the power calculation based on Wald test. Moreover, Moreover, the power can be improved to any desirable level by either increasing the number of clusters, $I$ or increasing the cluster size, $m$. Last but not least, an obvious optimal design (minimizing the required sample size given the power) is such that $\overline{W} = \theta = 0.5$, meaning clusters are being equally randomized to intervention and control arms. Under such an optimal design, the percentage of subjects within a cluster in one of two possible subgroups is 50\% for which the variance $\sigma^2_4(m, \rho, \rho_x, \overline{W}, \sigma^2_{y|x})$ reaches its minimal value $16\sigma^2_{\epsilon}/m$.

\subsection{Variable cluster size scenario}

So far, we have shown that condition~\eqref{fixed prevalence} helps remove $\rho$ (or $\sigma^2_{\gamma}$) from the expression of $\sigma^2_4(m, \rho, \rho_x, \overline{W}, \sigma^2_{y|x})$ according to \eqref{existing result}. Next, we wondered if the same result holds for the general cases accounting unequal cluster sizes. Unfortunately, by slightly modifying the condition \eqref{fixed prevalence} as
$$\mathbf{1}^{\top}_m\boldsymbol{X}_i = k = \theta m_i,\quad \theta\in(0, 1)$$
we cannot simply replicate the heuristic analysis directly from \eqref{existing result general} for 1-dimensional $\boldsymbol{\beta}_4$ like we did in Section~\ref{sec: main equal} because by imposing the above condition we do not have the same $\rho_x$ across all clusters. Nevertheless, we may consider a slightly different version of \eqref{existing result general}:
\begin{align*}
\sigma^2_4(F, \bar{m}, \rho, \rho_x, \overline{W}, \sigma^2_{y|x}) &= \frac{\sigma^2_{y|x}(1 - \rho)}{\sigma^2_x\overline{W}(1 - \overline{W})\mathbb{E}\left[m_i - \cfrac{\left(1 - \rho_x(m_i)\right)m_i\rho}{1 + (m_i - 1)\rho} - \cfrac{\rho_x(m_i)m^2_i\rho}{1 + (m_i - 1)\rho}\right]} \\
&= \frac{\sigma^2_{y|x}(1 - \rho)}{\sigma^2_x\overline{W}(1 - \overline{W})\mathbb{E}\left[m_i - \cfrac{\left[1 + \left(1 - \rho_x(m_i)\right)\right]m_i\rho}{1 + (m_i - 1)\rho}\right]}.
\end{align*}
Hence if
$$\rho_x(m_i) = -\frac{1}{m_i - 1}$$
for all $i=1,\ldots,I$, which is the case by fixing the proportion $\theta$ within each cluster, then we immediately have
$$\sigma^2_4(F, \bar{m}, \rho, \rho_x, \overline{W}, \sigma^2_{y|x}) = \frac{\sigma^2_{y|x}(1 - \rho)}{\sigma^2_x\overline{W}(1 - \overline{W})\mathbb{E}[m_i]} = \frac{\sigma^2_{y|x}(1 - \rho)}{\bar{m}\sigma^2_x\overline{W}(1 - \overline{W})},$$
which is exactly the right-hand side of equation~\eqref{fixed prevalence}. Given the above heuristic discussions, we may have the following Theorem~\ref{thm1} (see Appendix~\ref{app1: proof} for the proof):
\begin{theorem}\label{thm1}
Suppose
\begin{equation}\label{fixed prevalence general}
\mathbf{X}^{\top}_i\mathbf{1}_{m_i} = m_i\boldsymbol{\theta},
\end{equation}
where $\boldsymbol{\theta} = \begin{bmatrix} \theta_1 & \cdots & \theta_p\end{bmatrix}^{\top}$ satisfying $0 < \theta_l < 1$ for all $l=1,\ldots,p$ and $\mathbf{1}^{\top}_p\boldsymbol{\theta} < 1$. Then we have
\begin{equation}\label{main result}
\mathbb{V}ar(\hat{\boldsymbol{\beta}}_4|\{(W_i, \mathbf{X}_i)\}^I_{i=1}) = \frac{\sigma^2_{\epsilon}}{I\bar{m}\overline{W}_{\boldsymbol{m}}(1 - \overline{W}_{\boldsymbol{m}})}\left(\text{diag}(\boldsymbol{\theta}) - \boldsymbol{\theta}\boldsymbol{\theta}^{\top}\right)^{-1}.
\end{equation}
where $\text{diag}(\boldsymbol{\theta})\in[0, 1)^{p\times p}$ is a diagonal matrix with the elements of $\boldsymbol{\theta}$ in the main diagonal, 
$$\bar{m} = \frac{1}{I}\sum^I_{i=1}m_i = \frac{\boldsymbol{m}^{\top}\mathbf{1}_I}{I},\quad \overline{W}_{\boldsymbol{m}} = \frac{1}{I\bar{m}}\sum^I_{i=1}m_iW_i = \frac{\boldsymbol{m}^{\top}\boldsymbol{W}}{\boldsymbol{m}^{\top}\mathbf{1}_I}$$
and thus we should be able to write
$$\mathbb{V}ar(\hat{\boldsymbol{\beta}}_4) = \frac{\boldsymbol{\Omega}_4(\boldsymbol{m}, \boldsymbol{\theta}, \sigma^2_{\epsilon}, \mathbb{P})}{I}$$
due to the unbiasness of the GLS estimator $\hat{\boldsymbol{\beta}}_4$, where
\begin{equation}\label{result general}
\boldsymbol{\Omega}_4(\boldsymbol{m}, \boldsymbol{\theta}, \sigma^2_{\epsilon}, \mathbb{P}) = \frac{\sigma^2_{\epsilon}}{\bar{m}}\psi(\mathbb{P}; \boldsymbol{m})\left(\text{diag}(\boldsymbol{\theta}) - \boldsymbol{\theta}\boldsymbol{\theta}^{\top}\right)^{-1}
\end{equation}
with
$$\psi(\mathbb{P}; \boldsymbol{m}) = \mathbb{E}\left[\frac{1}{\overline{W}_{\boldsymbol{m}}(1 - \overline{W}_{\boldsymbol{m}})}\right]$$
under \eqref{fixed prevalence general}.
\end{theorem}
\begin{remark}
The result \eqref{main result} is an exact equation rather than an asymptotic or approximate result, i.e. it holds even if $I$ is small.
\end{remark}
\begin{remark}
Using the Sherman–Morrison formula, we have more explicitly that
$$\left(\text{diag}(\boldsymbol{\theta}) - \boldsymbol{\theta}\boldsymbol{\theta}^{\top}\right)^{-1} = \text{diag}(\boldsymbol{\theta})^{-1} + \frac{\mathbf{J}_p}{1 - \mathbf{1}^{\top}_p\boldsymbol{\theta}}.$$
\end{remark}
\begin{remark}
Note that there are $2^I - 2$ possible values of $\boldsymbol{W}\in\{0, 1\}^I$ (with the situations with all clusters in either the control $\overline{W}_{\boldsymbol{m}} = 0$ or intervention arm $\overline{W}_{\boldsymbol{m}} = 1$ excluded) and hence there are at most $2^{I-1} - 1$ possible values for
$$\frac{1}{\overline{W}_{\boldsymbol{m}}(1 - \overline{W}_{\boldsymbol{m}})}.$$
The probability measure $\mathbb{P}$, representing the randomization method, can be characterized by a vector of probabilities with at most $2^{I-1}$ dimension. Hence given the number of clusters $I$, the expectation $\psi(\mathbb{P}; \boldsymbol{m}) < \infty$ due to the finite support of $\overline{W}_{\boldsymbol{m}}$. Moreover, $\psi$ satisfies
\begin{equation}\label{invariant}
\psi(\mathbb{P}; \boldsymbol{m}) = \psi(\mathbb{P}; \lambda\boldsymbol{m})
\end{equation}
for all $\lambda > 0$, which implies that $\psi$ relies on only the relative proportions of the cluster sizes $\boldsymbol{m}$. Hence if $\lambda\in(0, 1)$ represents the attrition rate, it will only affect $\bar{m}$. Thus, 
$$\psi(\mathbb{P}; \boldsymbol{m})\left(\text{diag}(\boldsymbol{\theta}) - \boldsymbol{\theta}\boldsymbol{\theta}^{\top}\right)^{-1}$$
can be seen as the design effect for parallel CRTs under \eqref{fixed prevalence general}.
\end{remark}
\begin{remark}
Under the null hypothesis $H_0:\boldsymbol{\beta}_4 = \mathbf{0}$, a test statistic under \eqref{fixed prevalence general} is asymptotically (see supplementary A for details) given by
$$\frac{I\bar{m}\hat{\boldsymbol{\beta}}^{\top}_4\left(\text{diag}(\boldsymbol{\theta}) - \boldsymbol{\theta}\boldsymbol{\theta}^{\top}\right)\hat{\boldsymbol{\beta}}_4}{\psi(\mathbb{P}; \boldsymbol{m})\sigma^2_{\epsilon}}\sim \chi^2(p).$$
Hence the power can be calculated as
$$\int^{\infty}_{\chi^2_{1 - \alpha}(p)}f(x; p, \vartheta)dx$$
where $\chi^2_{1 - \alpha}(p)$ is the $1-\alpha$ quantile of the $\chi^2(p)$ distribution, $f(x; p, \vartheta)$ is the probability density function of the $\chi^2(p, \vartheta)$ distribution with noncentrality parameter
$$\vartheta = \frac{I\bar{m}\boldsymbol{\beta}^{\top}_4\left(\text{diag}(\boldsymbol{\theta}) - \boldsymbol{\theta}\boldsymbol{\theta}^{\top}\right)\boldsymbol{\beta}_4}{\psi(\mathbb{P}; \boldsymbol{m})\sigma^2_{\epsilon}}.$$
\end{remark}
\begin{remark}
Notice that the variance of the GLS of $\beta_2$ is given by
$$\mathbb{V}ar(\hat{\beta}_2) = \mathbb{E}\left[\frac{1}{\overline{W}_{\boldsymbol{m}}(\rho)\left(1 - \overline{W}_{\boldsymbol{m}}(\rho)\right)}\right]\frac{1}{\displaystyle\sum^I_{i=1}\cfrac{1}{\cfrac{\sigma^2_{\epsilon}}{m_i} + \sigma^2_{\gamma}}}$$
(resulting from the lower-right element of matrix $\sigma^2_{y|x}\mathbf{A}_{\boldsymbol{m}}(\rho)^{-1}$; see Appendix~\ref{app1: proof} for the definitions of $\mathbf{A}_{\boldsymbol{m}}(\rho)$ and $\overline{W}_{\boldsymbol{m}}(\rho)$), we have
$$\mathbb{E}\left[\frac{1}{\overline{W}_{\boldsymbol{m}}(\rho)\left(1 - \overline{W}_{\boldsymbol{m}}(\rho)\right)}\right]\frac{\sigma^2_{\gamma}}{I} < \mathbb{V}ar(\hat{\beta}_2) < \mathbb{E}\left[\frac{1}{\overline{W}_{\boldsymbol{m}}(\rho)\left(1 - \overline{W}_{\boldsymbol{m}}(\rho)\right)}\right]\frac{\sigma^2_{y|x}}{I}.$$
In contrast to $\mathbb{V}ar(\hat{\beta}_2)$, for which we have
$$\mathbb{V}ar(\hat{\beta}_2)\to 0$$
when $I\to\infty$ but
$$\varliminf_{\min\{m_1,\ldots,m_I\}\to\infty}\mathbb{V}ar(\hat{\beta}_2) \ge \frac{\sigma^2_{\gamma}}{I}\varliminf_{\min\{m_1,\ldots,m_I\}\to\infty}\mathbb{E}\left[\frac{1}{\overline{W}_{\boldsymbol{m}}(\rho)\left(1 - \overline{W}_{\boldsymbol{m}}(\rho)\right)}\right] \ge \frac{4\sigma^2_{\gamma}}{I} > 0,$$
we have by Theorem 1 that
$$\mathbb{V}ar(\hat{\boldsymbol{\beta}}_4)\to \mathbf{0}$$
when either $I\to\infty$ or $\bar{m}\to\infty$ given that the limit of $\psi(\mathbb{P}; \boldsymbol{m})$ exists, which indicates that for null hypothesis $H_0:\boldsymbol{\beta}_4 = \mathbf{0}$ of the HTE the power of the Wald test can be arbitrarily large by increasing either the number of clusters $I$ or the cluster sizes $\boldsymbol{m}$. Hence if the current ICC-ignorable parallel CRT for detecting HTE has insufficient power while expanding the number of clusters is difficult or costly, we may improve power through increasing the cluster sizes rather than the number of clusters.
\end{remark}
From Theorem~\ref{thm1} we can see that neither $\rho$ nor $\sigma^2_{\gamma}$ appears in \eqref{main result} or \eqref{result general}, i.e. the cluster random effect is completely irrelevant to the power/sample size calculation based on Wald tests. Note that only $\boldsymbol{m}$, $\theta$, and $\sigma^2_{\epsilon}$ are included as arguments of $\boldsymbol{\Omega}_4$, however, the randomization of the parallel CRT also affects the value of $\boldsymbol{\Omega}_4(\boldsymbol{m}, \boldsymbol{\theta}, \sigma^2_{\epsilon}, \mathbb{P})$. In the next section we shall discuss the component $ \psi(\mathbb{P}; \boldsymbol{m})$ under the random allocation rule and its approximation.

\subsection{Random allocation rule}

The random allocation rule specifies that $I_1$ out of $I$ clusters and $I_0$ out of $I$ clusters are in the intervention and control arms, respectively (usually, $I_1 = I_0 = I/2$) \citep{lachin1988properties}. Although not mentioned explicitly, this type of randomization is widely-used in many studies and protocols.\citep{mitja2021cluster,wolfe2009school,ogedegbe2014cluster}. As the simplest possible case of the permuted-block randomization, it is recommended for trials with $I<100$ in order to avoid treatment imbalances \citep{lachin1988randomization}.

Under the random allocation rule, if $\boldsymbol{m} = \bar{m}\mathbf{1}_I$, then 
$$\overline{W}_{\boldsymbol{m}} \equiv \frac{I_1}{I}.$$
Thus,
$$\mathbb{V}ar(\hat{\beta}_4) = \mathbb{V}ar(\hat{\beta}_4|\{(W_i, \boldsymbol{X}_i)\}^I_{i=1}) = \frac{I\sigma^2_{\epsilon}}{\bar{m}\theta(1 - \theta)I_0I_1}$$
under \eqref{fixed prevalence}. 

For the cases with unequal cluster sizes, under \eqref{fixed prevalence general} there are 
$$\displaystyle{I\choose I_1}=\displaystyle{I\choose I_0}$$
different combinations of $W_1,\ldots,W_I$ satisfying
$$\sum^I_{i=1}W_i = I_1,$$
which leads to
$$\psi(\mathbb{P}; \boldsymbol{m}) = \frac{1}{\displaystyle{I\choose I_1}}\sum_{\boldsymbol{W}\in\mathcal{W}_I}\frac{1}{\overline{W}_{\boldsymbol{m}}(1 - \overline{W}_{\boldsymbol{m}})},$$
where
$$\mathcal{W}_I = \left\{\boldsymbol{W}\in\{0, 1\}^I: \sum^I_{i=1}W_i = I_1\right\}.$$
For small values of $n$, $\psi(\mathbb{P}; \boldsymbol{m})$ can be computed using \emph{combn} function or \emph{comboGeneral} function of R package ``RcppAlgos''. However, as $n$ increases this method will be ultimately limited by the RAM allowance. Hence we may need to approximate this quantity for parallel CRTs with moderate and large number of clusters $I$. To this end, we first introduce the following Proposition~\ref{General moments} (see supplementary B for the proof).
\begin{proposition}\label{General moments}
Assume the random allocation rule is employed by a parallel CRT such that $I_1$ out of all $I$ ($I \ge 4$) clusters are assigned to the intervention arm. Then the 2nd central moment (variance) of $\overline{W}_{\boldsymbol{m}}$ is given by
\begin{equation}\label{General 2nd}
\mathbb{V}ar(\overline{W}_{\boldsymbol{m}}) = \frac{I_1I_0}{I^2(I - 1)}\text{CV}^2_{F_I}
\end{equation}
and
the 4th central moment of $\overline{W}_{\boldsymbol{m}}$ is given by
\begin{equation}\label{General 4th simplify}
\mathbb{E}\left[\left(\overline{W}_{\boldsymbol{m}} - \frac{I_1}{I}\right)^4\right] = \left(\frac{I^2 - 6I_0I_1 + I}{n}\text{Kurt}_{F_I} + 3(I_1 - 1)(I_0 - 1)\right)\frac{I_1I_0\text{CV}^4_{F_I}}{I^3(I - 1)(I - 2)(I - 3)},
\end{equation}
where $\text{CV}_{F_I}$ is the empirical coefficient of variation of $\boldsymbol{m}$ defined via
$$\text{CV}^2_{F_I} = \frac{1}{I\bar{m}^2}\sum^I_{i=1}m^2_i - 1$$
and $\text{Kurt}_{F_I}$ is the empirical Kurtosis of $\boldsymbol{m}$ defined via
$$\text{Kurt}_{F_I} = \frac{\cfrac{1}{I}\left(\displaystyle\sum^I_{i=1}m^4_i - 4\bar{m}\displaystyle\sum^I_{i=1}m^3_i + 6\bar{m}^2\displaystyle\sum^I_{i=1}m^2_i - 3\bar{m}^4\right)}{\left(\cfrac{1}{I}\displaystyle\sum^I_{i=1}m^2_i - \bar{m}^2\right)^2}.$$
\end{proposition}
Proposition~\ref{General moments} allows us to approximate $\psi(\mathbb{P}; \boldsymbol{m})$ through the geometric series. To see this, let us consider the typical situation with $I = 2I_1 = 2I_0$. In this case, we have
$$\mathbb{E}[\overline{W}_{\boldsymbol{m}}] = \frac{1}{I\bar{m}}\sum^I_{i=1}m_i\mathbb{E}[W_i] = \frac{1}{2}$$
and hence the variance \eqref{General 2nd} reduces to
$$\mathbb{V}ar(\overline{W}_{\boldsymbol{m}}) = \frac{\text{CV}^2_{F_I}}{4(I - 1)}$$
and the 4th central moment \eqref{General 4th simplify} reduces to
$$\mathbb{E}\left[\left(\overline{W}_{\boldsymbol{m}} - \frac{1}{2}\right)^4\right] = \frac{[3(I - 2) - 2\text{Kurt}_{F_I}]}{16I(I - 1)(I - 3)}\text{CV}^4_{F_I}.$$
Now if we denote
$$C_{\boldsymbol{m}} := \overline{W}_{\boldsymbol{m}} - \frac{1}{2}\in\left(-\frac{1}{2}, \frac{1}{2}\right),$$
then
$$\frac{1}{\overline{W}_{\boldsymbol{m}}(1 - \overline{W}_{\boldsymbol{m}})} = \frac{1}{\left(0.5 + C_{\boldsymbol{m}}\right)(0.5 - C_{\boldsymbol{m}})} = \frac{1}{0.25 - C^2_{\boldsymbol{m}}} = 4\sum^{\infty}_{k=0}4^kC^{2k}_{\boldsymbol{m}}.$$
Thus, under the random allocation rule, $\psi(\mathbb{P}; \boldsymbol{m})$ can be approximated by
\begin{align}\label{approx}
\psi(\mathbb{P}; \boldsymbol{m}) & \approx 4 + 16\mathbb{E}\left[C^2_{\boldsymbol{m}}\right] + 64\mathbb{E}\left[C^4_{\boldsymbol{m}}\right] = 4 + 16\mathbb{V}ar(\overline{W}_{\boldsymbol{m}}) + 64\mathbb{E}\left[\left(\overline{W}_{\boldsymbol{m}} - \frac{1}{2}\right)^4\right] \nonumber\\
&= 4\left(1 + \frac{\text{CV}^2_{F_I}}{I - 1} + \frac{[3(I - 2) - 2\text{Kurt}_{F_I}]\text{CV}^4_{F_I}}{I(I - 1)(I - 3)}\right).
\end{align}
\begin{remark}
Note that Eq.~\eqref{approx} indicates that $\psi(\mathbb{P}; \boldsymbol{m})\to 4$ as $I\to\infty$ when $\displaystyle\varlimsup_{I\to\infty}\text{CV}^2_{F_I} < \infty$, for which the asymptotic normality of $\hat{\boldsymbol{\beta}}_4$ (Proposition 1 in supplementary A) holds.
\end{remark}

\section{Simulation studies}\label{sec4: simulation}

Simulation studies in this section focus on two goals: first is to verify that the ICC $\rho$ is irrelavant to the power/sample size calculation; second is to evaluate the performance of approximation of $\psi(\mathbb{P}; \boldsymbol{m})$ given by \eqref{approx} in terms of the operating characteristics. Throughout this section, we assume patients are recruited from 8$q$ (where $q$ is a positive integer) clusters with the cluster sizes given by
$$\boldsymbol{m} = \mathbf{1}_{q}\otimes\begin{bmatrix} \dfrac{\bar{m}}{2} & \dfrac{\bar{m}}{2} & \dfrac{\bar{m}}{2} & \dfrac{\bar{m}}{2} & \bar{m} & \dfrac{5\bar{m}}{2} & 2\bar{m} & \dfrac{\bar{m}}{2} \end{bmatrix}^{\top}$$
(where $\otimes$ denotes the Kronecker product) in order to explore the performance of \eqref{approx} by varying the number of clusters. Moreover, we fix $p=1$, $\beta_1 = 0.15$, $\beta_2 = 0.25$, $\beta_3 = 0.1$, and $\sigma^2_{\epsilon} = 1$ in the data-generating process \eqref{model} while vary the ICC $\rho\in\{0.05, 0.5, 0.95\}$ to represent the weak, moderate, strong intra-class correlation, respectively. The HTE $\beta_4 = \Delta \in\{0.25, 0.35, 0.45\}$ is considered to illustrate the change of calculated power \citep{yang2020sample}. And $\bar{m}$ is used to represent the required sample size due to the invariant property provided by \eqref{invariant}. The covariates $\boldsymbol{W}$ and $\boldsymbol{X}_i$'s are simulated using the \emph{sample} function in R to fix the number of clusters in each arm as well as the number of both reference and target groups within each cluster. The simulated outcomes are fitted by the \emph{lme} function of R package ``nlme''. All analysis is implemented by R (version 4.4.1).

\subsection{The empirical SE vs. the computed SE}

We first evaluated the performance computed standard error (CSE) of $\hat{\beta}_4$ using Monte Carlo simulation. By setting $\theta = 0.5$, according to \eqref{result general} the CSE is given by
\begin{equation}\label{CSE}
\text{CSE} = \sqrt{\frac{4}{I\bar{m}}\psi(\mathbb{P}; \boldsymbol{m})}
\end{equation}
where $\psi(\mathbb{P}; \boldsymbol{m})$ is approximated by \eqref{approx}. By varying $\bar{m}\in\{20, 40, 60\}$ and $q\in\{1, 2, 3\}$ the resulting CSEs are shown in Table~\ref{tab1: standard deviation}. 

Next, we set $\beta_4 = 0.35$ and simulate 10,000 datasets through the data-generating process \eqref{model}. By fitting these 10,000 datasets with LMM we obtain 10,000 GLS estimations of $\beta_4$ which are denoted as $\hat{\beta}^{(1)}_4,\ldots,\hat{\beta}^{(10,000)}_4$. The empirical standard error (ESD), which is considered as the true standard deviation of $\beta_4$ \citep{yang2024power,ford2020maintaining}, is defined as
$$\text{ESD} = \sqrt{\frac{1}{9,999}\sum^{10,000}_{l=1}\left(\hat{\beta}^{(l)}_4 - \frac{1}{10,000}\sum^{10,000}_{l=1}\hat{\beta}^{(l)}_4\right)^2}.$$

Besides, for each simulated dataset the LMM provides an estimation of the standard error for the GLS estimator that is denoted as $\hat{\text{se}}\left(\hat{\beta}^{(l)}_4\right)$ for $l=1,\ldots,10,000$. Hence the average standard error
$$\overline{\text{SE}} = \frac{1}{10,000}\sum^{10,000}_{l=1}\hat{\text{se}}\left(\hat{\beta}^{(l)}_4\right)$$
is used to create performance measure such as the empirical SE percent bias \citep{ford2020maintaining}. The values of ESD and $\overline{\text{SE}}$ for each combination of $\bar{m}\in\{20, 40, 60\}$, $q\in\{1, 2, 3\}$, and $\rho\in\{0.05, 0.5, 0.95\}$ are also presented in Table~\ref{tab1: standard deviation}  (similar simulation studies for verifying Eq.~\eqref{main result} can be found in Supplementary C).

From Table~\ref{tab1: standard deviation}, the ICC $\rho$ plays no role in either the simulated ESD or the simulated $\overline{\text{SE}}$. Moreover, the CSD is almost identical to the average estimated standard error fitted by the LMM particularly when either $\bar{m}$ or $I = 8q$ is sufficiently large. Hence our proposed method to approximate the standard deviation of $\hat{\beta}_4$ based on \eqref{approx} shall share similar performance with those fitted by the LMM.

\subsection{Pre-randomization operating characteristics}

In this section, we the required average sample sizes corresponding to prespecified powers as well as the resulting empirical type-I errors and powers from simulation in ICC-ignorable parallel CRTs. Throughout this section, we fix the significance level $\alpha = 0.05$ and required power $0.8$ to calculate the minimal required average cluster size $\bar{m}$ by assuming 8 clusters ($q = 1$). Since
$$\mathbb{V}ar(\hat{\beta}_4) = \frac{4.380022}{8\bar{m}\theta(1 - \theta)} = \frac{0.5475028}{\bar{m}\theta(1 - \theta)}$$
according to \eqref{fixed prevalence general} and \eqref{approx}, the predicted power of test $H_0: \beta_4 = 0$ vs. $H_1: \beta_4 = \Delta$ is calculated as
\begin{equation}\label{power prerandomization}
\phi = \Phi\left(z_{0.025} + \lvert\Delta\rvert\sqrt{\frac{\bar{m}\theta(1 - \theta)}{0.5475028}}\right)
\end{equation}
based on the Wald test and thus the required average cluster size is
\begin{equation}\label{average size: PREP}
\bar{m}\ge \frac{0.5475028\left(z_{0.975} + z_{0.8}\right)^2}{\theta(1 - \theta)\Delta^2}
\end{equation}
where $z_*$ is the *-quantile of the standard normal distribution.

To ensure the proportion of the targe group within each cluster is exactly $\theta$, we choose $\bar{m}$ as multiples of 20, 10, and 4 that are closest to the value of
$$\frac{0.5475028\left(z_{0.975} + z_{0.8}\right)^2}{\theta(1 - \theta)\Delta^2}$$
corresponding to $\theta = 0.3$, $0.4$, and $0.5$ respectively. The resulting $\bar{m}$ is then used to calculate the predicted power $\phi$. The values of $(\bar{m}, \phi)$ for each combination of $\theta\in\{0.3, 0.4, 0.5\}$ and $\Delta\in\{0.25, 0.35, 0.45\}$ are provided in Table~\ref{tab2: pre randomization}.

The empirical type-I error $\psi_0$ and empirical power $\phi_0$ are calculated as the proportion of significant GLS estimations of $\beta_4$ among 10,000 results based on the Monte Carlo simulation, i.e.
$$\psi_0, \phi_0 = \frac{1}{10,000}\sum^{10,000}_{l = 1}\mathbf{1}\left\{\lvert\hat{\beta}^{(l)}_4\rvert > z_{0.975}\hat{\text{se}}\left(\hat{\beta}^{(l)}_4\right)\right\},$$
where $\mathbf{1}\{\bullet\}$ is the indicator function. The only difference is that $\beta_4 = 0$ in the data-generating process \eqref{model} when computing the emprical type-I error $\psi_0$ while $\beta_4 = \Delta\in\{0.25, 0.35, 0.45\}$ in the data-generating process \eqref{model} when computing the empirical power $\phi_0$. The values of $\psi_0$ and $\phi_0$ for each combination of $\theta\in\{0.3, 0.4, 0.5\}$, $\Delta\in\{0.25, 0.35, 0.45\}$, and $\rho\in\{0.05, 0.5, 0.95\}$ are also presented in Table~\ref{tab2: pre randomization} (similar simulation studies for power and sample size calculations based on Eq.~\eqref{main result} can be found in Supplementary D).

Again, the ICC $\rho$ plays no role in either the empirical type-I error $\psi_0$ or the empirical power $\phi_0$ according to Table~\ref{tab2: pre randomization}. Moreover, the calculated average cluster size $\bar{m}$ based on our proposed methods \eqref{fixed prevalence general} and \eqref{approx} can roughly generate the desired empirical power $\psi_0$ which is closed to the predicted power $\phi$ for all scenarios listed in Table~\ref{tab2: pre randomization}. In addition, the empirical type-I error $\psi_0$ is well-controlled under all scenarios.

\subsection{Cluster sizes vs. number of clusters}

When ICC is ignorable in power calculation, the power can be increased to any desirable level by either increasing the number of clusters or the cluster sizes. In this section, we shall investigate how the required average sample size $\bar{m}$ varies according to the change of the number of clusters. By fixing $\theta = 0.5$, the predicted power is given by
\begin{equation}\label{power q}
\phi = \Phi\left(z_{0.025} + \frac{\lvert\Delta\rvert}{\text{CSE}}\right)
\end{equation}
where CSE is given by \eqref{CSE} and thus the required average cluster size is
\begin{equation}\label{average size: q}
\bar{m} \ge \frac{\psi(\mathbb{P};\boldsymbol{m})\left(z_{0.975} + z_{0.8}\right)^2}{2q\Delta^2}.
\end{equation}
To ensure the the targe group and reference group are with 1:1 ratio within each cluster, we choose $\bar{m}$ as multiples of 4 that are closest to the value of
$$\frac{\psi(\mathbb{P};\boldsymbol{m})\left(z_{0.975} + z_{0.8}\right)^2}{2q\Delta^2}.$$
The resulting $\bar{m}$ is then used to calculate the predicted power $\phi$ using \eqref{power q}. The values of $(\bar{m}, \phi)$ for each combination of $q\in\{2, 3, 4\}$ and $\Delta\in\{0.25, 0.35, 0.45\}$ are provided in Table~\ref{tab3: pre randomization mn}. The values of empirical type-I error $\psi_0$ and empirical power $\phi_0$ for each combination of $q\in\{2, 3, 4\}$, $\Delta\in\{0.25, 0.35, 0.45\}$, and $\rho\in\{0.05, 0.5, 0.95\}$ are also presented in Table~\ref{tab3: pre randomization mn}.

The ICC $\rho$ still plays no role in either the empirical type-I error $\psi_0$ or the empirical power $\phi_0$ according to Table~\ref{tab3: pre randomization mn}. Moreover, given the required power, the calculated average cluster size $\bar{m}$ and the number of clusters $I$ roughly satisfies 
\begin{equation}\label{equalizer}
I\bar{m} = \frac{\psi(\mathbb{P};\boldsymbol{m})\left(z_{0.975} + z_{0.8}\right)^2}{\Delta^2},
\end{equation}
although the total sample size tends to be slightly smaller for larger number of clusters. The results of empirical type-I error $\psi_0$ and empirical power $\phi_0$ are similar as in Table~\ref{tab2: pre randomization}.

\subsection{Drop-out scenario}

In this section, we shall investigate how the drop-out rate affect the empirical type-I error and power. Throughout this section, we fix the significance level $\alpha = 0.05$ and required power 0.8 to calculate the minimal required average cluster size $\bar{m}$ by assuming 8 clusters ($q = 1$) and $\theta = 0.5$. Denote the drop-out rate as $r\in (0, 1)$. Then the actual number of observations available is $I\bar{m}(1 - r)$. Now let $K$ be the number of actual observations from the target subgroup. Then $K$ is a hypergeometric $\left(I\bar{m}, I\bar{m}\theta, I\bar{m}(1 - r)\right)$ random variable such that
$$\mathbb{E}[K] = I\bar{m}(1 - r)\theta,\quad \mathbb{V}ar(K) = \frac{(I\bar{m})^2r(1 - r)\theta(1 - \theta)}{I\bar{m} - 1}\approx I\bar{m}r(1 - r)\theta(1 - \theta).$$
Next, we assume that given $K$ the numbers of actual observations from the target subgroup within each cluster $M_1,\ldots,M_I$ is a multinomial $\left(K, \dfrac{m_1}{I\bar{m}}, \ldots, \dfrac{m_I}{I\bar{m}}\right)$ random vector and the numbers of actual observations from the reference subgroup within each cluster $N_1,\ldots,N_I$ is a multinomial $\left(I\bar{m}(1 - r) - K, \dfrac{m_1}{I\bar{m}}, \ldots, \dfrac{m_I}{I\bar{m}}\right)$ random vector. Then the required average cluster size $\bar{m}$ could be derived from
$$\bar{m} \ge \frac{1}{I}\sum^I_{i=1}\mathbb{E}\left[\frac{0.5475028\left(z_{0.975} + z_{0.8}\right)^2(M_i + N_i)^2}{(1 - r)M_iN_i\Delta^2}\right]$$
instead of \eqref{average size: PREP} to account for variation of $\theta$ across all clusters due to the drop-out. In the end, $\bar{m}$ can be approximately obtained by solving the following equation:
\begin{equation}\label{average size: drop-out}
\bar{m} = \frac{0.5475028\left(z_{0.975} + z_{0.8}\right)^2}{(1 - r)\Delta^2}\left\{\frac{1}{\theta(1 - \theta)} + \frac{\theta^3 + (1 - \theta)^3}{I\bar{m}(1 - r)(1 - \theta)^2\theta^2}\left[r + \frac{1}{I}\sum^I_{i=1}\frac{I\bar{m} - m_i}{m_i}\right]\right\}
\end{equation}
(see supplementary E for details). In our current settings, the above equation of $\bar{m}$ reduces to
\begin{equation}\label{average size: drop-out this case}
\bar{m} = \frac{0.5475028\left(z_{0.975} + z_{0.8}\right)^2}{(1 - r)\Delta^2}\left[4 + \frac{r + 5.45}{2\bar{m}(1 - r)}\right].
\end{equation}
Note that the equation \eqref{average size: drop-out this case} of $\bar{m}$ has a unique positive solution, with which we may roughly predict the power as
\begin{equation}\label{power: drop-out}
\phi = \Phi\left(z_{0.025} + \lvert\Delta\rvert\sqrt{\frac{(1 - r)\bar{m}}{0.5475028}\left[4 + \frac{r + 5.45}{2\bar{m}(1 - r)}\right]^{-1}}\right)
\end{equation}
based on the Wald test. To ensure the number of available observations is an integer, we choose $\bar{m}$ as multiples of 10, 4, and 10 that are closest to the unique positive solution to \eqref{power: drop-out}corresponding to scenarios $r=0.2$, $0.25$, and $0.3$. Hence it remains to examine how the ICC may affect the empirical power and type-I error of our proposed ICC-ignorable design with drop-out rate $r$. To simulate the real situation with drop-out rate $r$, the values of $\psi_0$ and $\varphi_0$ for each combination of $r\in\{0.2, 0.25, 0.3\}$, $\Delta\in\{0.25, 0.35, 0.45\}$, and $\rho\in\{0.05, 0.5, 0.95\}$ are also presented in Table~\ref{tab4: drop-out}.

According to Table~\ref{tab4: drop-out}, the ICC $\rho$ does not substantially affect the empirical power or type-I error. Moreover, the predicted power given by \eqref{power: drop-out} provides reasonable approximation for the actual power for all scenarios.

\section{Applications}\label{sec5: application}

This section provides examples of power and sample size calculation for parallel CRTs that assume ignorable ICC using data from three actual clinical trials: the RECODE trial \citep{kruis2013recode}, the PARTNER study \citep{PARTNER}, and the EPIC study \citep{EPIC}.

\subsection{RECODE Trial: Minimum Detectable HTE for 80\% power}\label{subsec: application1}

The RECODE trial is a two-arm, multi-center, pragmatic CRT investigating the effectiveness of integrated disease management (IDM) on the quality of life of patients with chronic obstructive pulmonary disease (COPD) \citep{kruis2014effectiveness}. This study randomly assigned 1,086 patients from $I = 40$ general practices to IDM (intervention) or usual care (control). The randomization used matching with variables such as percentage of minority patients, type of practice, practice location, age, and sex of practitioner. This randomization led to 1:1 ratio of intervention and control ($I_0 = I_1 = 20$). For illustration purpose, we ignored the matching steps because we cannot include the covariates in the data generating LMM in equation~\eqref{model}. Moreover, we simplify the permuted block randomization to the random allocation rule to illustrate how the power is calculated for HTE detection without the knowledge of the ICC.

The study analyzed the Clinical COPD Questionnaire in patients with COPD, finding a standard deviation $\sigma_{\epsilon} = 0.49$ at 12 months. A planned subgroup analysis examined the HTE based on Medical Research Council (MRC) dyspnea score. To detect HTE between MRC groups (1-2 versus 3-5: 66.6\% vs. 33.4\%)\citep{kruis2013recode}, the study required 1,080 participants for 80\% accounting for loss-to-follow-up and using a minimal clinically important difference (MCID) of -0.4 for the CCQ \citep{kruis2013recode}. However, the true HTE ($\Delta$) and the ICC $\rho$ were unknown. By applying our ICC-ignorable design, we were able to describe the power $\varphi$ as a function of the absolute value of $\Delta$. According to Table 3 of \citet{kruis2014effectiveness}, we may roughly set $\theta = 1 / 3$ for our ICC-ignorable design. Under the equal cluster size $\bar{m} = 27$ for the $40$ clusters, the power function under our ICC-ignorable design, is given by
$$\phi(\lvert\Delta\rvert) = \Phi\left(z_{0.025} + \frac{\lvert\Delta\rvert\sqrt{60}}{0.49}\right),$$
which is depicted by the solid line in Figure~\ref{Fig1a}. 

\begin{figure*}[!h]%
    \centering
    \begin{subfigure}{0.5\textwidth}
        \centering
        \includegraphics[width=.99\textwidth,height=.77\textwidth]{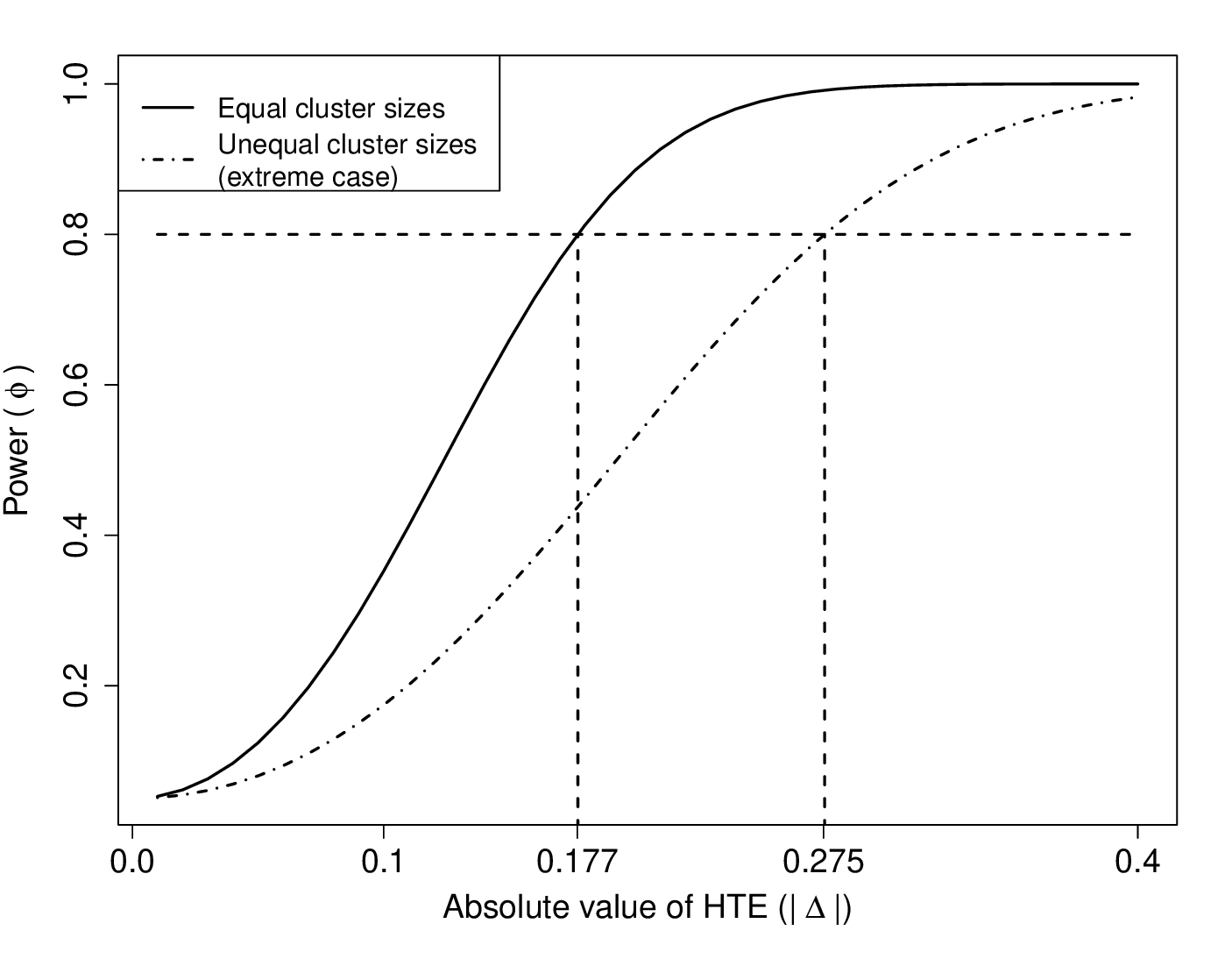}
        \caption{The power as a function of the absolute value of the HTE, given the average cluster size $\bar{m}=27$.}\label{Fig1a}
    \end{subfigure}%
    \begin{subfigure}{0.5\textwidth}
        \centering
        \includegraphics[width=.99\textwidth,height=.77\textwidth]{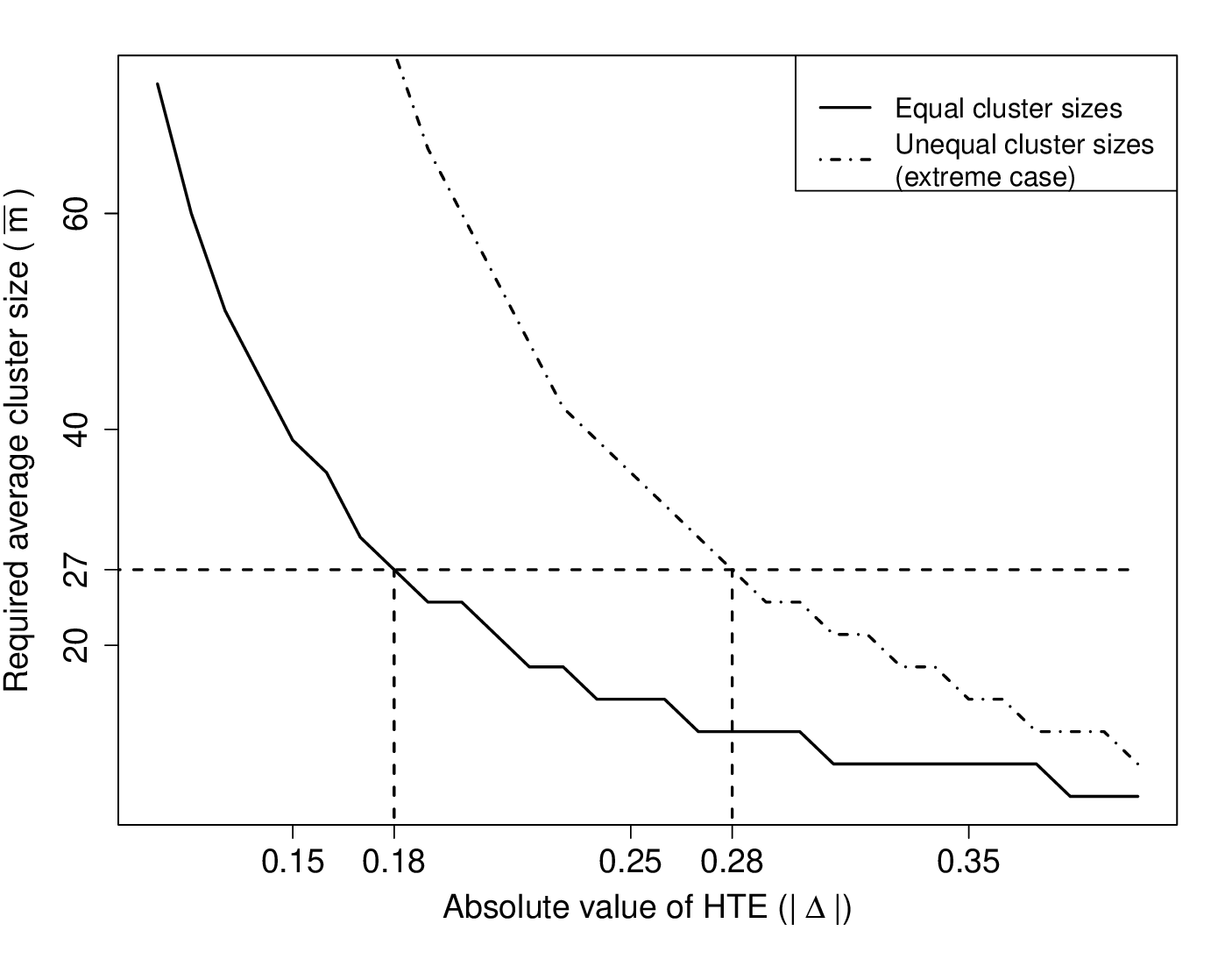}
        \caption{The required average cluster size as a function of the absolute value of the HTE given $80\%$ power.}\label{Fig1b}
    \end{subfigure}
    \caption{The power and required average cluster size of ICC-ignorable parallel CRT design with number of clusters.}
\end{figure*}

Note that the case of equal cluster size is an optimal ICC ignorable design in the sense that the power is maximized given the total sample size 1,080. Hence, we also consider it an extreme case with $m_1 = \cdots = m_{39} = 3$ patients while $m_{40} = 963$ patients. The extreme case leads to $\psi(\mathbb{P}; \boldsymbol{m})\approx 9.6577$ due to \eqref{approx} and thus
$$\phi(\lvert\Delta\rvert) \approx \Phi\left(z_{0.025} + \frac{\lvert\Delta\rvert\sqrt{24.8507}}{0.49}\right),$$
which is depicted by the dashed line in Figure~\ref{Fig1a}. Therefore, under the ICC-ignorable parallel CRT design with random allocation rule, if the HTE absolute value of interest $\lvert\Delta\rvert < 0.177$ then the sample size 1,080 is insufficient for achieving 80\% power; if the HTE absolute value $\lvert\Delta\rvert \in [0.177, 0.275)$, then 80\% power is achievable by the sample size 1,080 under weak imbalance of cluster sizes; if the HTE absolute value $\lvert\Delta\rvert > 0.275$ then the total sample size 1,080 is guaranteed to achieve 80\% power regardless of the cluster size imbalance.

In Figure~\ref{Fig1b}, the solid line depicts the required average cluster size $\bar{m}$ for the case of equal cluster sizes as a function of $\lvert\Delta\rvert$, which is actually
$$\bar{m}(\lvert\Delta\rvert) = \left\lceil\frac{9\left(z_{0.975} + z_{0.8}\right)^2(0.49)^2}{20\Delta^2}\cdot \frac{1}{3}\right\rceil\cdot 3,$$
where $\lceil\bullet\rceil$ is the ceiling function while the dashed line depicts the required average cluster size  $\bar{m}$ for the extreme case as a function of $\lvert\Delta\rvert$. To obtain the dashed line in Figure~\ref{Fig1b}, one has to solve \eqref{equalizer} for $m_{40}$ by fixing $m_1 = \cdots = m_{39} = 3$ because $m_{40}$ appears in both sides of \eqref{equalizer}. 

Figure~\ref{Fig1b} again indicates that under the ICC-ignorable parallel CRT design with random allocation rule, the total sample size 1,080 is insufficient to achieve 80\% power if the HTE absolute value of interest $\lvert\Delta\rvert <0.18$; if the HTE absolute value $\lvert\Delta\rvert \in [0.18, 0.28)$, the 80\% power might be achievable depending on the cluster size imbalance; if the HTE absolute value $\lvert\Delta\rvert \ge 0.28$, then 80\% power is guaranteed for the total sample size 1,080 regardless of the variation among cluster sizes.

In summary, the RECODE trial \citep{kruis2013recode} involving approximately 1,080 COPD patients as an ICC-ignorable CRT allows for the detection of HTE of IDM between MRC groups with 80\% statistical power, provided the HTE is at least 0.18. Therefore, if the MRC subgroup analysis of the RECODE trial yield a significant result smaller than 0.18, it may be attributed to insufficient power and would necessitate an increase in sample size. For HTE between 0.18 and 0.28, while a lack of power remains a possibility, redistributing cluster sizes might be more effective than increasing sample size. However, for HTE greater than 0.28, the trial is guaranteed to achieve 80\% statistical power regardless of the distribution of cluster sizes.

\subsection{PARTNER study: Improved Trial Design Using Random Allocation Rule}\label{subsec: application2}

The PARTNER study \citep{PARTNER} employed a two-arm cluster-randomized trial to evaluate the impact of a telephone-based asthma coaching program for parents on three key outcomes: symptom-free days for children, parental quality-of-life, and emergency department visits. The intervention aimed to enhance primary care management for children with persistent asthma. A total of 948 eligible families were recruited from $I = 22$ community-based primary care practices that provided asthma care to at least 40 children. Stratum-specific randomization was used at the practice level to ensure an equal allocation ($I_0 = I_1 = 11$). For the purpose illustration, we assume the resulting intervention and control arms were from the random allocation rule. 

Apart from the main analysis, HTE between insurance types (Medicaid vs. other) were examined. The primary outcome (change in PACQLQ score) was assumed to have standard deviation 0.91, which allowed us to evaluate the power as a function of the HTE under the ICC-ignorable parallel CRT design with random allocation rule. As there were 223 out of 948 families having Medicaid, we set $\theta = 1 / 4$ in our power and sample size calculation. Hence under the equal cluster size $\bar{m} = (420 + 463) / 22 \approx 40$ at the end of the 12 months for the 22 clusters, the power function under our ICC-ignorable design is given by
$$\phi(\lvert\Delta\rvert) = \Phi\left(z_{0.025} + \frac{\lvert\Delta\rvert\sqrt{165}}{2(0.91)}\right),$$
which is depicted by the solid lines in Figure~\ref{Fig2a}. 

\begin{figure*}[!h]%
    \centering
    \begin{subfigure}{0.5\textwidth}
        \centering
        \includegraphics[width=.99\textwidth,height=.77\textwidth]{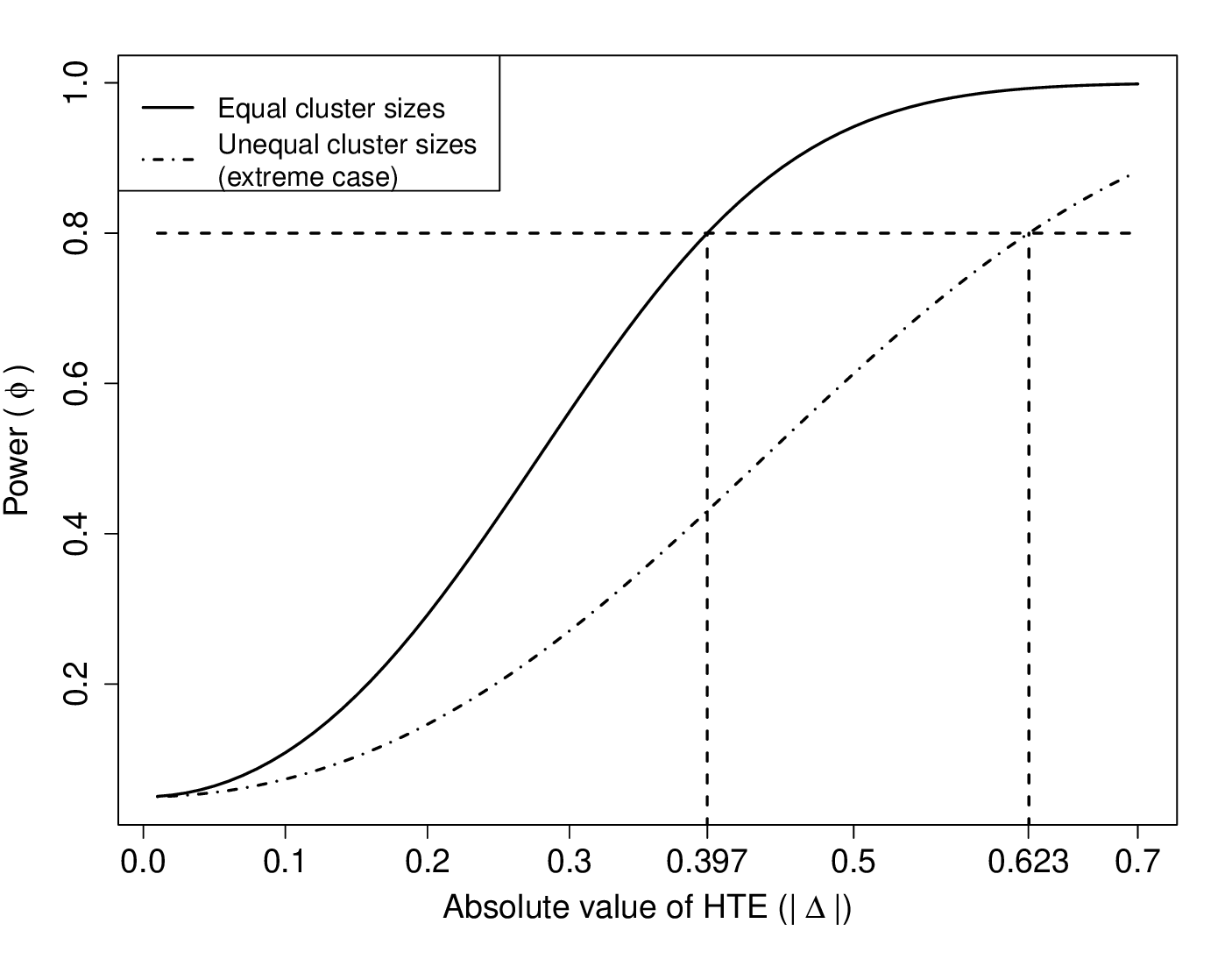}
        \caption{The power as a function of the absolute value of the HTE, given the average cluster size $\bar{m}=40$.}\label{Fig2a}
    \end{subfigure}%
    \begin{subfigure}{0.5\textwidth}
        \centering
        \includegraphics[width=.99\textwidth,height=.77\textwidth]{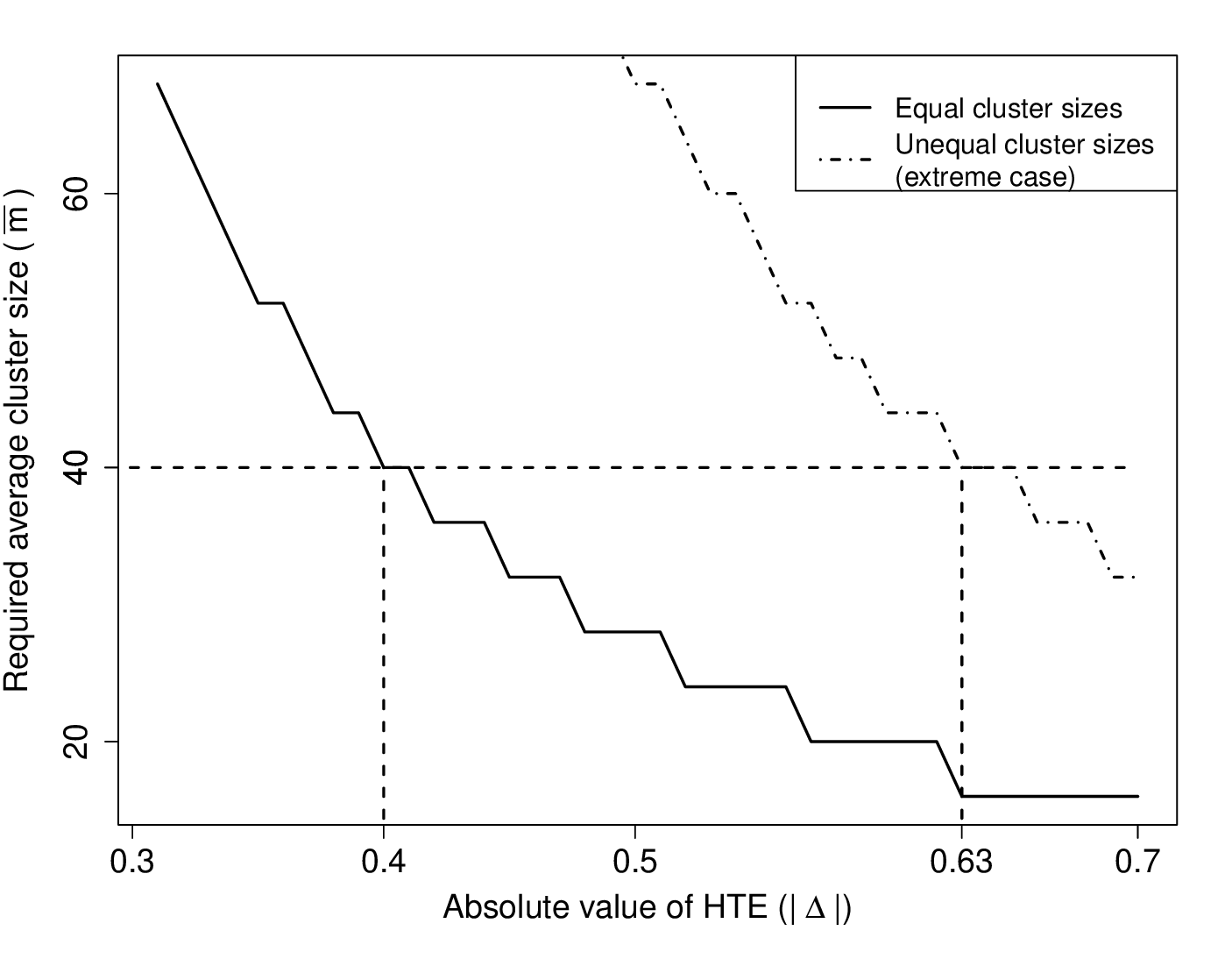}
        \caption{The required average cluster size as a function of the absolute value of the HTE given $80\%$ power.}\label{Fig2b}
    \end{subfigure}
    \caption{The power and required average cluster size of ICC-ignorable parallel CRT design with number of clusters.}
\end{figure*}

We also consider an extreme case with $m_1 = \cdots = m_{21} = 4$ patients while $m_{22} = 796$ patients. The extreme case leads to $\psi(\mathbb{P}; \boldsymbol{m})\approx 9.8644$ due to \eqref{approx} and thus
$$\phi(\lvert\Delta\rvert) \approx \Phi\left(z_{0.025} + \frac{\lvert\Delta\rvert\sqrt{16.7268}}{0.91}\right),$$
which is depicted by the dashed line in Figure~\ref{Fig2a}. Therefore, under the ICC-ignorable parallel CRT design with random allocation rule, if the HTE absolute value of interest $\lvert\Delta\rvert < 0.397$ then the sample size 880 is insufficient for achieving 80\% power; if the HTE absolute value $\lvert\Delta\rvert\in[0.397, 0.623)$, then 80\% power is achievable by the sample size 880 under weak imbalance of cluster sizes; if the HTE absolute value $\lvert\Delta\rvert > 0.623$ then the total sample size 880 is guaranteed to achieve 80\% power regardless of the cluster size imbalance.

In Figure~\ref{Fig2b}, the solid line depicts the required average cluster size $\bar{m}$ for the case of equal cluster sizes as a function of $\lvert\Delta\rvert$, which is
$$\bar{m}(\lvert\Delta\rvert) = \left\lceil\frac{64\left(z_{0.975} + z_{0.8}\right)^2(0.91)^2}{3(22)\Delta^2}\cdot\frac{1}{4}\right\rceil\cdot 4.$$
To obtain the dashed line in Figure~\ref{Fig2b}, one has to solve \eqref{equalizer} for $m_22$ by fixing $m_1 = \cdots = m_{21} = 4$ because $m_{22}$ appears in both sides of \eqref{equalizer}.

Figure~\ref{Fig2b} again indicates that under the ICC-ignorable parallel CRT design with random allocation rule, the total sample size 880 is insufficient to achieve 80\% power if the HTE absolute value of interest $\lvert\Delta\rvert < 0.40$; if the HTE absolute value of interest $\lvert\Delta\rvert\in[0.40, 0.63)$, the 80\% power might be achievable depending on the cluster size imbalance; if the HTE absolute value $\lvert\Delta\rvert \ge 0.63$, then 80\% power is guaranteed for the total sample size 880 regardless of the variation among cluster sizes.

The subgroup analysis at 24 months suggest an HTE near 0.5. With even cluster size,  the sample size (880) may suffice for 80\%. This likely holds true if practices with fewer than 40 asthma patients were included. Our ICC-ignorable design only requires a 1:3 Medicaid: other insurance ratio among practices. Including these practices could reveal a sufficiently powered, significant HTE between Medicaid and other insured families.

\subsection{EPIC study: Operating Characteristics under drop-out}\label{subsec: application3}

EPIC study \citep{EPIC}, utilized a two-arm cluster-randomized trial design to assess the impact of two interventions on the health and wellness of older adults residing in low-income independent housing in Los Angeles. The interventions were 1) health and wellness classes focused on improving physical and mental well-being and 2) in-home preventive healthcare program. A total of 480 participants from $I = 16$ low-income independent older adult apartment buildings were randomly assigned to one of the two intervention arms in a 1:1 ratio ($I_0 = I_1 = 8$). In addition, the proportion of those receiving the health and wellness classes was set to 
$$\overline{W}_{\boldsymbol{m}} = \frac{240}{480} = \frac{1}{2}$$
regardless of the randomization, which leads to $\psi(\mathbb{P}; \boldsymbol{m}) = 4$. Hence under the assumption of 25\% drop-out, if 480 is considered as the worst-case scenario sample size, then the planned sample size should be 640 \citep{EPIC}. The primary outcomes (change in total score, physical functioning subscale, and mental functioning subscales of PROMIS 10-item Global Health Scale) were assumed to have standard error $\sigma_{\epsilon} = 10$. One subgroup analysis of interest is to examine the potential HTE between the subgroup with depression (GDS score 0-4) and the subgroup without depression (GDS score 5-15)\citep{greenberg2012geriatric}. Hence if the ICC-ignorable design is considered, the prevalence of depression $\theta$ can be roughly set as 0.25 for low-income older adults \citep{choi2009unmet}, leading to the required average sample size as a function of $\lvert\Delta\rvert$ given by solving the following equation inspired by \eqref{average size: drop-out}:
$$\bar{m}(\lvert\Delta\rvert) = \frac{\left(z_{0.975} + z_{0.8}\right)^2(10)^2}{3\Delta^2}\left\{\frac{16}{3} + \frac{15.25(28)}{27\bar{m}(\lvert\Delta\rvert)}\right\}$$
(Figure~\ref{Fig3b}), and the power as a function of $\lvert\Delta\rvert$ given by
$$\phi(\lvert\Delta\rvert) = \Phi\left(z_{0.025} + \frac{\lvert\Delta\rvert}{10}\sqrt{120\left(\frac{16}{3} + \frac{15.25(28)}{1080}\right)^{-1}}\right)$$
(Figure~\ref{Fig3a}). Hence the ICC-ignorable CRT design with 16 clusters and average cluster size 30 (assuming 25\% dropouts from average cluster size 40) can detect $\lvert\Delta\rvert \ge 6.13$ with at least 80\% power. Now consider the ICC-ignorable CRT design with the same $\theta = 0.25$ and average cluster size 30 without any dropouts (called the ``No-dropout case''). The power as a function of $\lvert\Delta\rvert$ is then given by
$$\phi(\lvert\Delta\rvert) = \Phi\left(z_{0.025} + \frac{\lvert\Delta\rvert\sqrt{22.5}}{10}\right)$$
and the required average cluster size for achieving 80\% is
$$\bar{m}(\lvert\Delta\rvert) = \frac{4\left(z_{0.975} + z_{0.8}\right)^2(10)^2}{3\Delta^2}.$$
Figure~\ref{Fig3a} and Figure~\ref{Fig3b} also imply that for the No-dropout case the same average cluster size 30 can detect $\lvert\Delta\rvert\ge 5.91$ with 80\% power. This is slightly lower than that of the case described in EPIC study due to that the dropouts introduce extra variation to the cluster sizes.

\begin{figure*}[!h]%
    \centering
    \begin{subfigure}{0.5\textwidth}
        \centering
        \includegraphics[width=.99\textwidth,height=.77\textwidth]{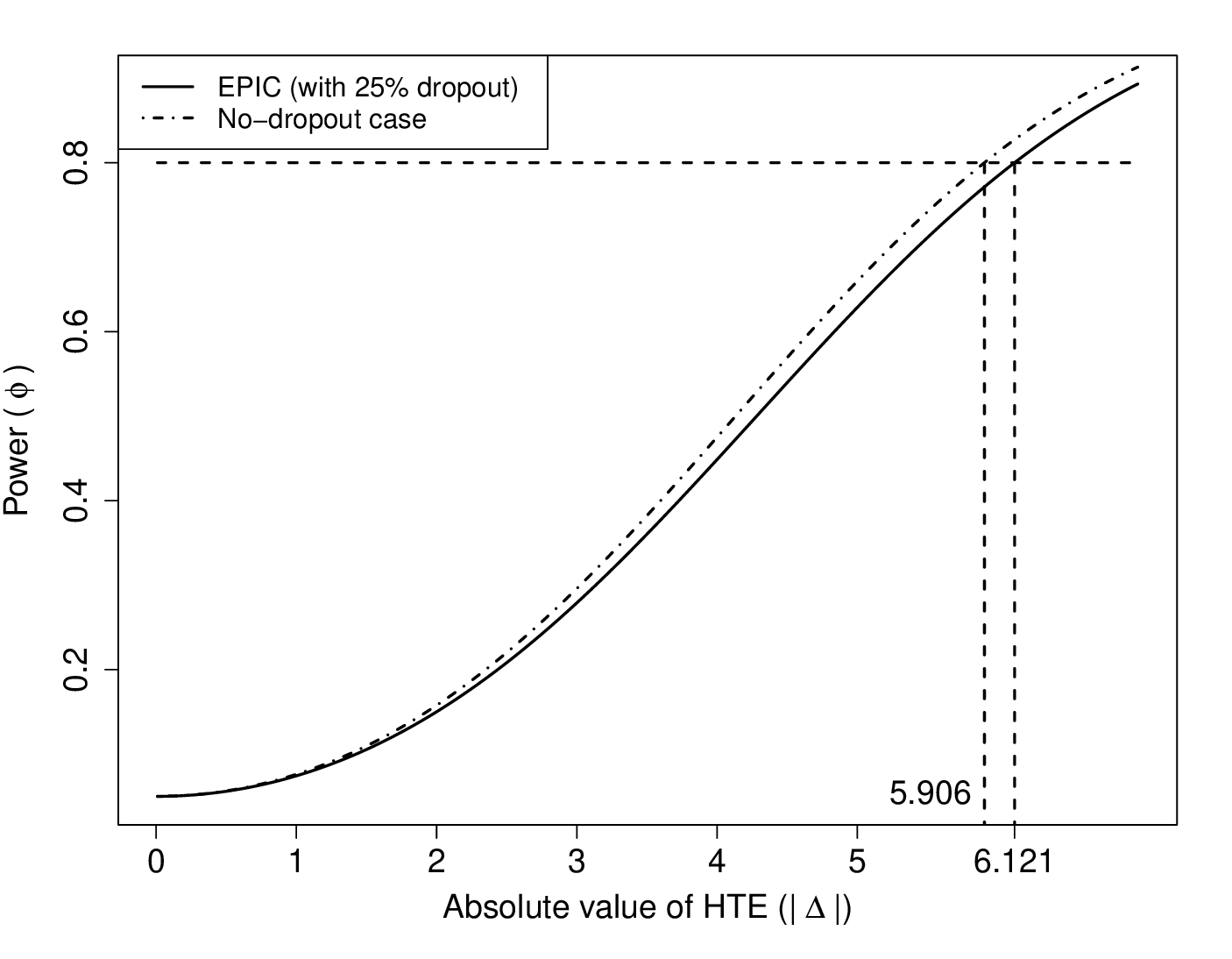}
        \caption{The power as a function of the absolute value of the HTE, given the average cluster size $\bar{m}=30$.}\label{Fig3a}
    \end{subfigure}%
    \begin{subfigure}{0.5\textwidth}
        \centering
        \includegraphics[width=.99\textwidth,height=.77\textwidth]{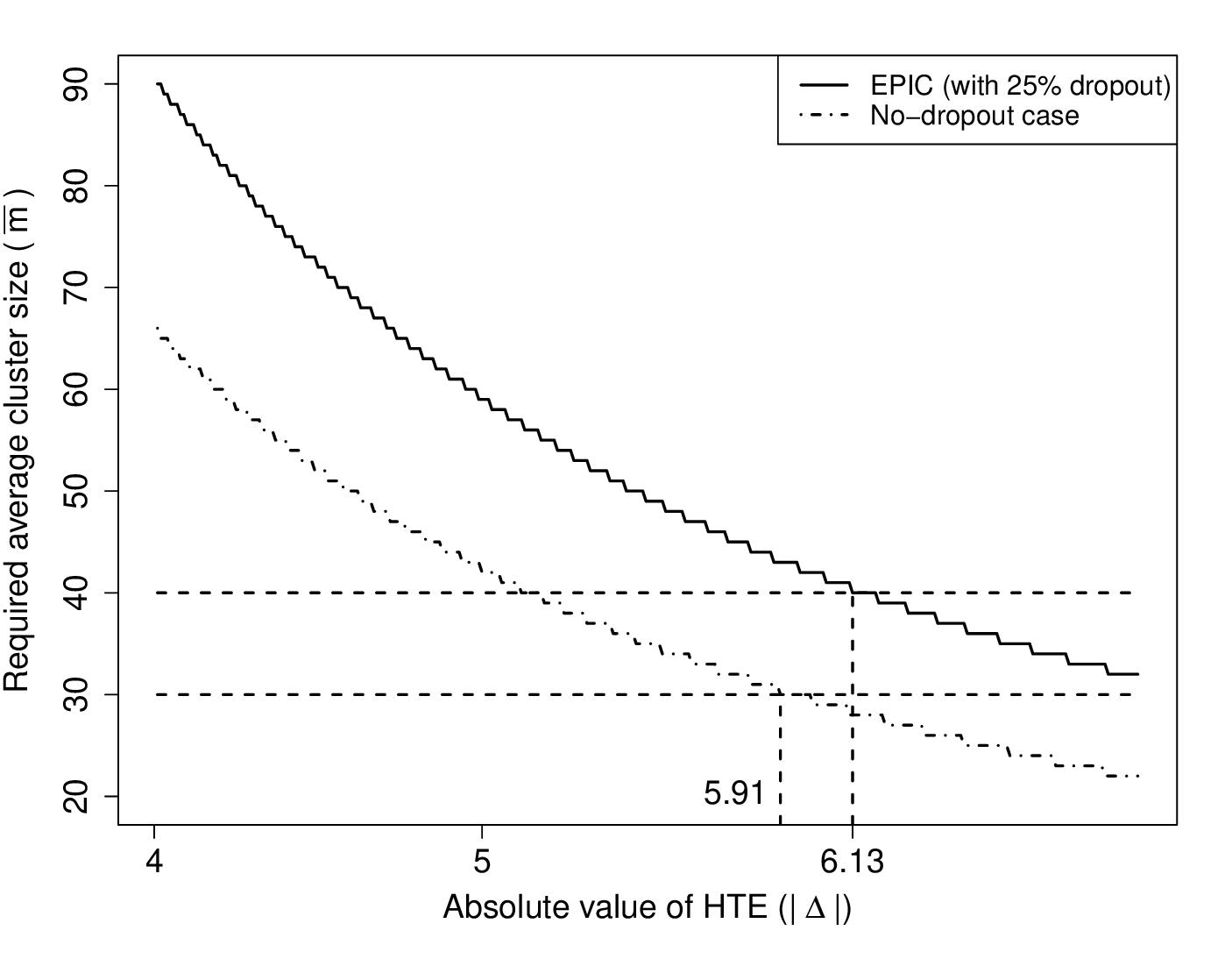}
        \caption{The required average cluster size as a function of the absolute value of the HTE given $80\%$ power.}\label{Fig3b}
    \end{subfigure}
    \caption{The power and required average cluster size of ICC-ignorable parallel CRT design for the depression subgroup analysis of the EPIC study.}
\end{figure*}

If there is no preliminary information about the prevalence of depression, then we shall keep it as $\theta\in(0, 1)$ and check the power as a function of $\theta$ given $\lvert\Delta\rvert$:
$$\phi(\theta; \lvert\Delta\rvert) = \Phi\left(z_{0.025} + \frac{\lvert\Delta\rvert}{10}\sqrt{120}\left\{\frac{1}{\theta(1 - \theta)} + \frac{15.25[\theta^3 + (1 - \theta)^3]}{480(1 - \theta)^2\theta^2}\right\}^{-1}\right),$$
and the required average sample size as a function of $\lvert\Delta\rvert$ given by solving the following equation:
$$\bar{m}(\theta; \lvert\Delta\rvert) = \frac{\left(z_{0.975} + z_{0.8}\right)^2(10)^2}{3\Delta^2}\left\{\frac{1}{\theta(1 - \theta)} + \frac{15.25[\theta^3 + (1 - \theta)^3]}{12\bar{m}(\theta; \lvert\Delta\rvert)(1 - \theta)^2\theta^2}\right\}.$$
Hence the power of ICC-ignorable CRT design with 16 clusters and average cluster size 30 (assuming 25\% dropouts from average cluster size 40) for $\lvert\Delta\rvert = 2.5$, $5$, $7.5$, and $10$ can be found in Figure~\ref{Fig4a}. Besides, Figure~\ref{Fig4b} depicts the required average cluster size to achieving 80\% power for $\lvert\Delta\rvert = 2.5$, $5$, $7.5$, and $10$.

\begin{figure*}[!h]%
    \centering
    \begin{subfigure}{0.5\textwidth}
        \centering
        \includegraphics[width=.99\textwidth,height=.77\textwidth]{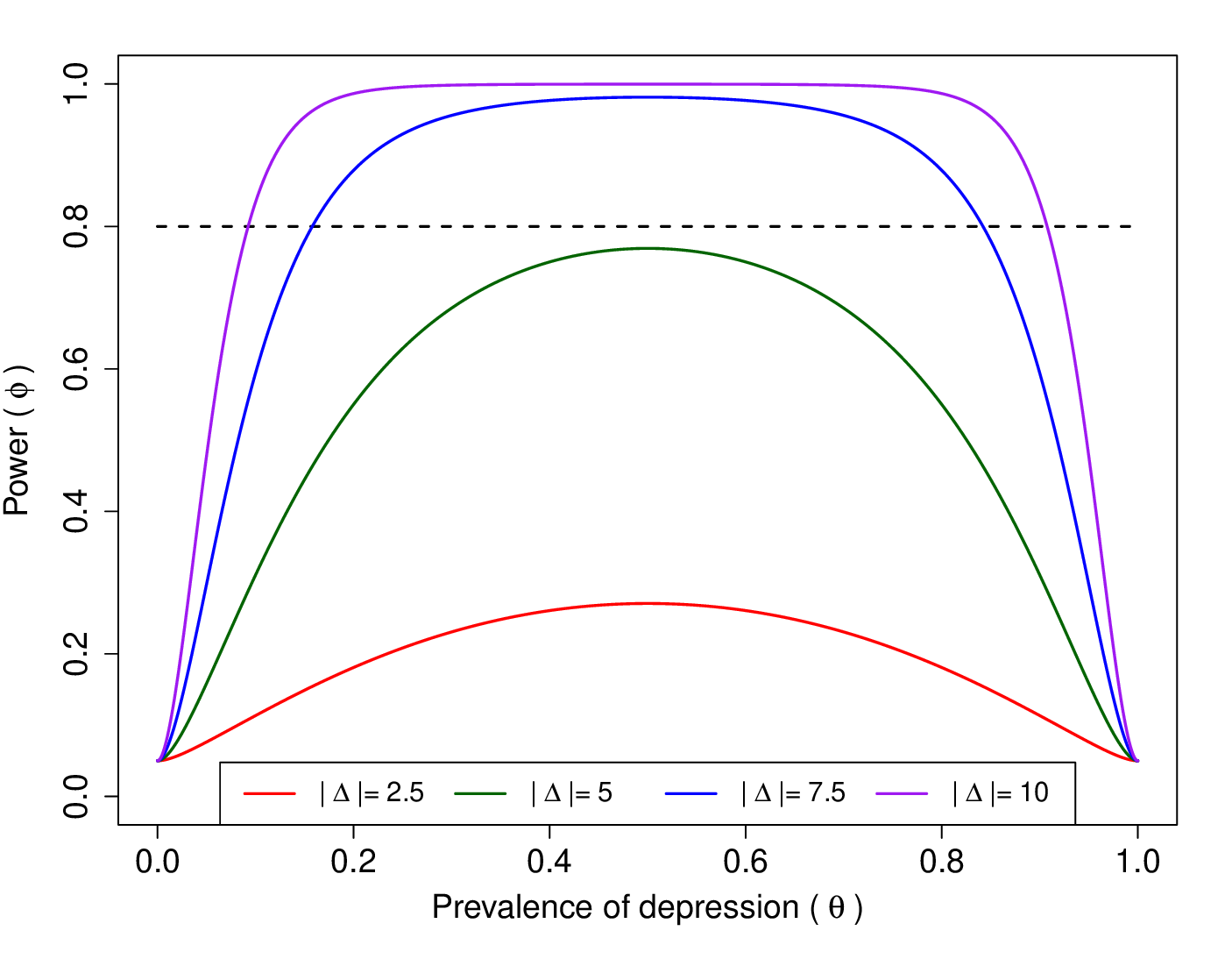}
        \caption{The power as a function of the absolute value of the HTE, given the average cluster size $\bar{m}=30$.}\label{Fig4a}
    \end{subfigure}%
    \begin{subfigure}{0.5\textwidth}
        \centering
        \includegraphics[width=.99\textwidth,height=.77\textwidth]{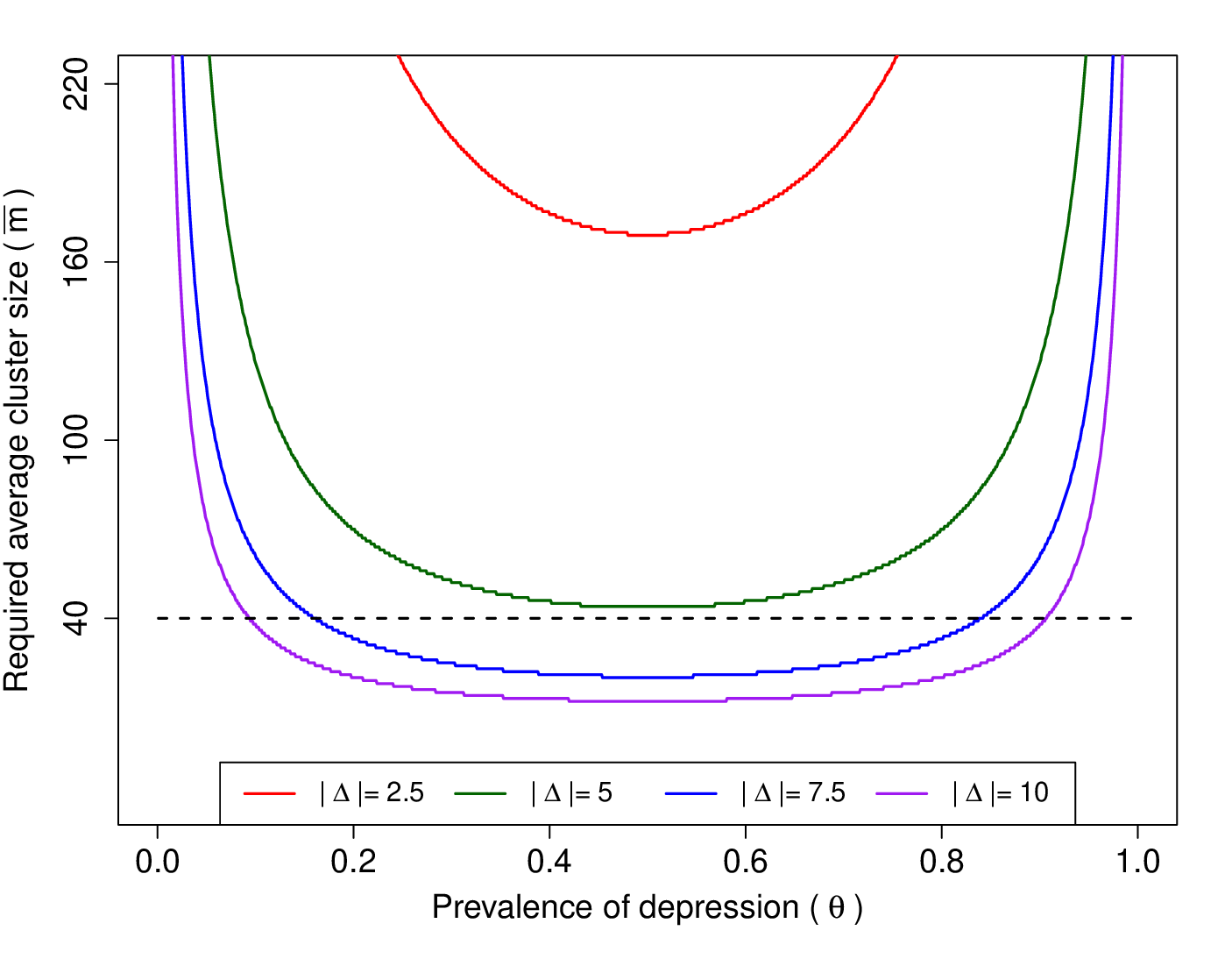}
        \caption{The required average cluster size as a function of the absolute value of the HTE given $80\%$ power.}\label{Fig4b}
    \end{subfigure}
    \caption{The power and required average cluster size of ICC-ignorable parallel CRT design for the depression subgroup analysis of the EPIC study.}
\end{figure*}

By taking the EPIC study as an example, we illustrate the power and sample size analysis based on the ICC-ignorable design when dropout rate is provided by the trial. According to our analysis, the EPIC study may detect HTE no less than 5.196 between depression groups with at least 80\% power using the ICC-ignorable design. If a significant result of HTE is smaller than 5.196 then the EPIC study may need to increase sample size. Moreover, the EPIC study can detect a HTE greater than 7.5 when the prevalence of depression varies between 0.158 and 0.842; greater than 10 when the prevalence of depression varies between 0.092 and 0.908.

\section{Discussion}\label{sec6: discussion}

In this paper, we identify a class of parallel CRTs for detecting the HTE, for which the ICC can be completely ignored from the power and sample size calculation. This class of parallel CRTs is characterized by fixed ratio between the size of the target group and reference group within each cluster. This  indicates that the power for detecting the HTE can be improved to any desirable level by increasing either the number of clusters or the size of each cluster. This flexibility allows for greater adaptability in trial design. For instance, the requirement in existing trials that primary care practices must provide asthma care to at least 40 or more children could be removed when using this approach to evaluate the effectiveness of telephone peer coaching for parents in reducing childhood asthma morbidity \citep{kruis2014effectiveness}. Section~\ref{subsec: application2} of the study demonstrates that primary care practices providing asthma care to children from four families can be included in the CRT as long as these families adhere to a 1:3 Medicaid-to-other insurance ratio. Relaxing these restrictions allow for the inclusion of more clusters in the CRT, which can enhance statistical power if necessary. Additionally, the ICC-ignorable property proves advantageous for trials when increasing the number of clusters is challenging or expensive.

The ENSPIRE study, which analyzed HTE across numerous covariates, aimed for 90\% power with a rigorous test of the ICC assumption (ranging from 0.01 to 0.05)\citep{hsu2024engaging}. To achieve this, they randomized 40 long-term care centers (LTCC) with thousands of frontline healthcare workers as clusters. If some HTE analyses with significant results lacked sufficient statistical power, an ICC-ignorable CRT design could be employed. This approach would enhance power by increasing the individuals within existing LTCC, thereby avoiding the need to recruit additional LTCC.

Finally, the ICC-ignorable CRT design is particularly suitable for subgroup analysis of CRTs when the ICC is unknown. This is exemplified by the ACTIVITAL study, which analyzed student subgroups based on weight categories (underweight, normal BMI, overweight/obese) without explicit ICC information \citep{andrade2016two}. In such cases, our methodology allows for direct power and sample size calculations without the need for ICC estimation, simplifying the planning and analysis process for these specific subgroup investigations.  be obtained directly from our result without considering the.

In addition to identifying the class of ICC-ignorable CRTs, this paper also provides theoretical results as well as a numerical approximation of the standard error of the GLS estimator of the HTE particularly under the random allocation rule. As 1:1 ratio between intervention and control arms is common for CRTs to maximize the power, we focus on the simplest case of randomization that guarantees the 1:1 ratio. The result reveals that under the random allocation rule the power of ICC-ignorable CRTs is substantially affected by the empirical CV of cluster sizes, which coincides with the result of Tong et al for general CRTs with unequal cluster sizes \citep{tong2022accounting}. In fact, the CV or imbalance of cluster sizes may lead to different results in power and sample size determination as illustrated in Section~\ref{subsec: application1} and Section~\ref{subsec: application2}. This implication suggests that when ICC-ignorable CRT is used, the power may be improved by reducing the cluster size imbalance if feasible before increasing the number of clusters or cluster sizes.

The impact of unequal cluster sizes in CRTs has been extensively studied \citep{tong2022accounting,pan2001sample,kerry2001unequal,van2007relative,eldridge2006sample}. For many studies, the coefficient of variation plays a pivotal role in determining the design effect \citep{ tong2022accounting,van2007relative}. One major benefit for this framework is that the information of the randomization of clusters with unequal cluster sizes is implicitly contained in the prior information F from which the cluster sizes are assumed to be drawn at random. However, the obvious drawback is that the cluster sizes might be very difficult to simulate given only the CV (see \citet{tong2022accounting} for example). For this reason, we propose an analytic framework to separate the randomization scheme from the cluster sizes. With the help of this framework, we obtained an approximation which contains only empirical CV and kurtosis of cluster sizes for the random allocation rule. For the simulation studies, the randomization can be implemented by directly using simple random sample. Under our proposed framework, it is possible to explore the power and sample size calculation under other randomization schemes, for example, the permuted block randomization which was actually used by \citet{kruis2014effectiveness}. Furthermore, if there is uncertainty about the prevalence of subgroups, which is expected, statisticians can still evaluate which plausible HTE values the study will have sufficient power to detect by plotting the power function as a function of prevalence for a given HTE value.

Potential limitations of this study include the restriction of continuous outcomes with data-generating process based on the LMM, and the strict condition to ensure the same proportion of participants from different subgroups for all clusters. In some situations, continuous subgroup identifiers might be more interesting \citep{yang2020sample,tong2022accounting}. Hence the ICC-ignorable property may not be applicable for detecting the HTE in such cases. Moreover, for binary and count outcomes, the ICC-ignorable property may not hold either due to the generalized linear mixed-effect data-generating process. Furthermore, the condition to ensure the same proportion of participants from different subgroups for all clusters may not be satisfied due to loss to follow-up. Nevertheless, intuitively the homogeneity of subgroup proportions among clusters can help reduce the impact of ICC on the power and sample size calculation. This should be easily verified by simulation studies.

In summary, this paper establishes a class of parallel cluster randomized trials (CRTs) where the intraclass correlation coefficient (ICC) becomes irrelevant for power and sample size calculations when the ratio between intervention and control groups within each cluster is fixed. This finding offers valuable flexibility in trial design, allowing for the inclusion of smaller clusters or a greater number of clusters, potentially expanding the pool of eligible participants and enhancing the feasibility of conducting CRTs.

Beyond identifying this class of ICC-ignorable CRTs, this study provides a theoretical framework for power and sample size calculations, particularly under the common scenario of 1:1 randomization. This framework highlights the critical influence of cluster size variability on trial power, emphasizing the importance of minimizing cluster size imbalance to optimize study efficiency.

While this work focuses on continuous outcomes and parallel CRTs, the concept of ICC-ignorability has broader implications. Future research will investigate its potential applicability in stepped-wedge CRTs, where both ICC and cluster autocorrelation play significant roles. Furthermore, exploring the robustness of ICC-ignorability to violations of the key condition of equal proportion of participants from target subgroup among all clusters and its applicability to binary and count outcomes will be crucial for expanding the practical utility of these findings.

%
%

\section*{Financial disclosure}

Dr. Mazumdar receives grant funding paid to her institution for grants related to this work from NCI and NCATS (R56CA267957, P30CA196521, U01TR002997-01A1, P30AG028741, UL1 TR004419) and unrelated to this work from NCI and NCATS (R35CA220491, U24CA224319, P01AG066605, U01CA121947). Dr. Kwon was partially supported by the support provided by the Biostatistics/Epidemiology/ Research Design (BERD) component of the Center for Clinical and Translational Sciences (CCTS) for this project that is currently funded through a grant (UL1TR003167), funded by the National Center for Advancing Translational Sciences (NCATS), awarded to the University of Texas Health Science Center at Houston.

\section*{Conflict of interest}

The authors declare no potential conflict of interests.

\bibliographystyle{plainnat}
\bibliography{ICC-Ignorability}

\pagebreak

\listoftables

\pagebreak

\begin{center}
\begin{table*}[!h]%
\caption{Simulation results for examining the performance of computed standard error (CSE) of the GLS estimator $\hat{\beta}_4$ given by \eqref{CSE} under various combinations of $\bar{m}$ and $q$. The ESD and $\overline{\text{SE}}$ are computed by fitting the LMM from 10,000 Monte Carlo simulations under various choices of the ICC $\rho$, compared with the CSE given by \eqref{CSE} with $\psi(\mathbb{P}; \boldsymbol{m})$ approximated by \eqref{approx}.\label{tab1: standard deviation}}
\begin{tabular*}{\textwidth}{@{\extracolsep\fill}lllllllllll@{}}
\toprule
& &\multicolumn{3}{@{}c}{$q = 1$} & \multicolumn{3}{@{}c}{$q = 2$} & \multicolumn{3}{@{}c}{$q= 3$} \\\cmidrule{3-5}\cmidrule{6-8}\cmidrule{9-11}
$\bar{m}$ & $\rho$ & ESD & $\overline{\text{SE}}$ & CSE & ESD & $\overline{\text{SE}}$ & CSE & ESD & $\overline{\text{SE}}$ & CSE \\
\midrule
\multirow{3}{*}{20} & 0.05 & 0.3315 & 0.3300 & \multirow{3}{*}{0.3309} & 0.2310 & 0.2278 & \multirow{3}{*}{0.2282} & 0.1829 & 0.1849 & \multirow{3}{*}{0.1849} \\
 & 0.50 & 0.3315 & 0.3305 & & 0.2310 & 0.2279 & & 0.1829 & 0.1849 & \\
 & 0.95 & 0.3315 & 0.3305 & & 0.2310 & 0.2279 & & 0.1829 & 0.1849 & \\
\midrule
\multirow{3}{*}{40} & 0.05 & 0.2354 & 0.2341 & \multirow{3}{*}{0.2340} & 0.1609 & 0.1612 & \multirow{3}{*}{0.1613} & 0.1311 & 0.1308 & \multirow{3}{*}{0.1308} \\
 & 0.50 & 0.2354 & 0.2341 & & 0.1609 & 0.1612 & & 0.1311 & 0.1308 & \\
 & 0.95 & 0.2354 & 0.2341 & & 0.1609 & 0.1612 & & 0.1311 & 0.1308 & \\
\midrule
\multirow{3}{*}{60} & 0.05 & 0.1939 & 0.1911 & \multirow{3}{*}{0.1911} & 0.1312 & 0.1317 & \multirow{3}{*}{0.1317} & 0.1076 & 0.1067 & \multirow{3}{*}{0.1068} \\
 & 0.50 & 0.1939 & 0.1912 & & 0.1312 & 0.1317 & & 0.1076 & 0.1067 & \\
 & 0.95 & 0.1939 & 0.1912 & & 0.1312 & 0.1317 & & 0.1076 & 0.1067 & \\
\bottomrule
\end{tabular*}
\end{table*}
\end{center}

\pagebreak

\begin{center}
\begin{table*}[!h]%
\caption{Simulation results for examining the performance of the predicted power (PREP) given by \eqref{power prerandomization} under various combinations of $\theta$ and $\Delta$. The estimated required average cluser size $\bar{m}$ is based on \eqref{average size: PREP}; the empirical Type-I error $\psi_0$ and the empirical power $\phi_0$ are computed by fitting the LMM from 10,000 Monte Carlo simulations under various choices of the ICC $\rho$; and the predicted power $\phi$ obtained from \eqref{power prerandomization}.\label{tab2: pre randomization}}
\begin{tiny}
\begin{tabular*}{\textwidth}{@{\extracolsep\fill}llllllllllllll@{}}
\toprule
& &\multicolumn{4}{@{}c}{$\Delta = 0.25$} & \multicolumn{4}{@{}c}{$\Delta = 0.35$} & \multicolumn{4}{@{}c}{$\Delta = 0.45$} \\\cmidrule{3-6}\cmidrule{7-10}\cmidrule{11-14}
$\theta$ & $\rho$ & $\bar{m}$ & $\psi_0$ & $\phi_0$ & $\phi$ & $\bar{m}$ & $\psi_0$ & $\phi_0$ & $\phi$ & $\bar{m}$ & $\psi_0$ & $\phi_0$ & $\phi$ \\
\midrule
\multirow{3}{*}{0.3} & 0.05 & \multirow{3}{*}{320} & 0.0472 & 0.7924 & \multirow{3}{*}{0.7910} & \multirow{3}{*}{160} & 0.0493 & 0.7754 & \multirow{3}{*}{0.7829} & \multirow{3}{*}{100} & 0.0487 & 0.7930 & \multirow{3}{*}{0.7959} \\
 & 0.50 & & 0.0472 & 0.7924 & & & 0.0493 & 0.7754 & & & 0.0487 & 0.7931 & \\
 & 0.95 & & 0.0472 & 0.7924 & & & 0.0493 & 0.7754 & & & 0.0487 & 0.7931 & \\
\midrule
\multirow{3}{*}{0.4} & 0.05 & \multirow{3}{*}{290} & 0.0483 & 0.8088 & \multirow{3}{*}{0.8048} & \multirow{3}{*}{150} & 0.0496 & 0.8048 & \multirow{3}{*}{0.8101} & \multirow{3}{*}{90} & 0.0482 & 0.8002 & \multirow{3}{*}{0.8069} \\
 & 0.50 & & 0.0483 & 0.8088 & & & 0.0496 & 0.8048 & & & 0.0480 & 0.8003 & \\
 & 0.95 & & 0.0483 & 0.8088 & & & 0.0496 & 0.8048 & & & 0.0480 & 0.8003 & \\
\midrule
\multirow{3}{*}{0.5} & 0.05 & \multirow{3}{*}{276} & 0.0541 & 0.8008 & \multirow{3}{*}{0.8014} & \multirow{3}{*}{140} & 0.0502 & 0.8045 & \multirow{3}{*}{0.7991} & \multirow{3}{*}{84} & 0.0495 & 0.7869 & \multirow{3}{*}{0.7959} \\
 & 0.50 & & 0.0541 & 0.8007 & & & 0.0502 & 0.8043 & & & 0.0494 & 0.7867 & \\
 & 0.95 & & 0.0541 & 0.8007 & & & 0.0502 & 0.8043 & & & 0.0494 & 0.7867 & \\
\bottomrule
\end{tabular*}
\end{tiny}
\end{table*}
\end{center}

\pagebreak

\begin{center}
\begin{table*}[!h]%
\caption{Simulation results for examining the performance of the predicted power (PREP) given by \eqref{power q} under various combinations of $q$ and $\Delta$. The estimated required average cluser size $\bar{m}$ is based on \eqref{average size: q}; the empirical Type-I error $\psi_0$ and the empirical power $\phi_0$ are computed by fitting the LMM from 10,000 Monte Carlo simulations under various choices of the ICC $\rho$; and the predicted power $\phi$ is obtained from \eqref{power q}.\label{tab3: pre randomization mn}}
\begin{tiny}
\begin{tabular*}{\textwidth}{@{\extracolsep\fill}llllllllllllll@{}}
\toprule
& &\multicolumn{4}{@{}c}{$\Delta = 0.25$} & \multicolumn{4}{@{}c}{$\Delta = 0.35$} & \multicolumn{4}{@{}c}{$\Delta = 0.45$} \\\cmidrule{3-6}\cmidrule{7-10}\cmidrule{11-14}
$q$ & $\rho$ & $\bar{m}$ & $\psi_0$ & $\phi_0$ & $\phi$ & $\bar{m}$ & $\psi_0$ & $\phi_0$ & $\phi$ & $\bar{m}$ & $\psi_0$ & $\phi_0$ & $\phi$ \\
\midrule
\multirow{3}{*}{2 (i.e. $I = 16$)} & 0.05 & \multirow{3}{*}{132} & 0.0505 & 0.8011 & \multirow{3}{*}{0.8037} & \multirow{3}{*}{68} & 0.0524 & 0.8193 & \multirow{3}{*}{0.8074} & \multirow{3}{*}{40} & 0.0489 & 0.7953 & \multirow{3}{*}{0.7965} \\
 & 0.50 & & 0.0505 & 0.8011 & & & 0.0525 & 0.8193 & & & 0.0488 & 0.7954 & \\
 & 0.95 & & 0.0505 & 0.8011 & & & 0.0525 & 0.8194 & & & 0.0488 & 0.7955 & \\
\midrule
\multirow{3}{*}{3 (i.e. $I = 24$)} & 0.05 & \multirow{3}{*}{86} & 0.0538 & 0.8051 & \multirow{3}{*}{0.8004} & \multirow{3}{*}{44} & 0.0495 & 0.7992 & \multirow{3}{*}{0.8015} & \multirow{3}{*}{28} & 0.0509 & 0.8239 & \multirow{3}{*}{0.8210} \\
 & 0.50 & & 0.0536 & 0.8053 & & & 0.0495 & 0.7992 & & & 0.0508 & 0.8240 & \\
 & 0.95 & & 0.0537 & 0.8052 & & & 0.0495 & 0.7992 & & & 0.0508 & 0.8240 & \\
\midrule
\multirow{3}{*}{4 (i.e. $I = 32$)} & 0.05 & \multirow{3}{*}{64} & 0.0494 & 0.7983 & \multirow{3}{*}{0.8001} & \multirow{3}{*}{32} & 0.0466 & 0.7875 & \multirow{3}{*}{0.7921} & \multirow{3}{*}{20} & 0.0453 & 0.8127 & \multirow{3}{*}{0.8049} \\
 & 0.50 & & 0.0493 & 0.7985 & & & 0.0465 & 0.7875 & & & 0.0453 & 0.8125 & \\
 & 0.95 & & 0.0493 & 0.7984 & & & 0.0465 & 0.7876 & & & 0.0453 & 0.8127 & \\
\bottomrule
\end{tabular*}
\end{tiny}
\end{table*}
\end{center}

\pagebreak

\begin{center}
\begin{table*}[!h]%
\caption{Simulation results for examining the performance of the predicted power (PREP) given by \eqref{power: drop-out} under various combinations of $r$ and $\Delta$. The estimated required average cluser size $\bar{m}$ is based on \eqref{average size: drop-out this case}; the empirical Type-I error $\psi_0$ and the empirical power $\phi_0$ are computed by fitting the LMM from 10,000 Monte Carlo simulations under various choices of the ICC $\rho$; and the predicted power $\phi$ is obtained from \eqref{power: drop-out}.\label{tab4: drop-out}}
\begin{tiny}
\begin{tabular*}{\textwidth}{@{\extracolsep\fill}llllllllllllll@{}}
\toprule
& &\multicolumn{4}{@{}c}{$\Delta = 0.25$} & \multicolumn{4}{@{}c}{$\Delta = 0.35$} & \multicolumn{4}{@{}c}{$\Delta = 0.45$} \\\cmidrule{3-6}\cmidrule{7-10}\cmidrule{11-14}
$r$ & $\rho$ & $\bar{m}$ & $\psi_0$ & $\phi_0$ & $\phi$ & $\bar{m}$ & $\psi_0$ & $\phi_0$ & $\phi$ & $\bar{m}$ & $\psi_0$ & $\phi_0$ & $\phi$ \\
\midrule
\multirow{3}{*}{0.2} & 0.05 & \multirow{3}{*}{340} & 0.0501 & 0.7918 & \multirow{3}{*}{0.7936} & \multirow{3}{*}{180} & 0.0470 & 0.8053 & \multirow{3}{*}{0.8064} & \multirow{3}{*}{110} & 0.0498 & 0.8061 & \multirow{3}{*}{0.8079} \\
 & 0.50 & & 0.0504 & 0.7914 & & & 0.0471 & 0.8048 & & & 0.0496 & 0.8066 & \\
 & 0.95 & & 0.0505 & 0.7915 & & & 0.0474 & 0.8056 & & & 0.0495 & 0.8062 & \\
\midrule
\multirow{3}{*}{0.25} & 0.05 & \multirow{3}{*}{368} & 0.0458 & 0.8043 & \multirow{3}{*}{0.7994} & \multirow{3}{*}{188} & 0.0512 & 0.7992 & \multirow{3}{*}{0.7980} & \multirow{3}{*}{116} & 0.0506 & 0.8051 & \multirow{3}{*}{0.8034} \\
 & 0.50 & & 0.0454 & 0.8043 & & & 0.0510 & 0.7994 & & & 0.0503 & 0.8042 & \\
 & 0.95 & & 0.0454 & 0.8043 & & & 0.0510 & 0.7991 & & & 0.0506 & 0.8046 & \\
\midrule
\multirow{3}{*}{0.3} & 0.05 & \multirow{3}{*}{400} & 0.0490 & 0.8078 & \multirow{3}{*}{0.8051} & \multirow{3}{*}{200} & 0.0482 & 0.8001 & \multirow{3}{*}{0.7952} & \multirow{3}{*}{120} & 0.0492 & 0.7997 & \multirow{3}{*}{0.7893} \\
 & 0.50 & & 0.0485 & 0.8086 & & & 0.0482 & 0.8002 & & & 0.0492 & 0.7992 & \\
 & 0.95 & & 0.0484 & 0.8051 & & & 0.0481 & 0.8002 & & & 0.0493 & 0.7997 & \\
\bottomrule
\end{tabular*}
\end{tiny}
\end{table*}
\end{center}

\pagebreak

\section*{Supporting information}

Additional supporting information may be found in the online version of the article at the publisher’s website.

\appendix

\section{Proofs}\label{app1: proof}
\vspace*{12pt}
\begin{proof}[Proof of Theorem~\ref{thm1}]
For $\boldsymbol{Y}_i$ generated by \eqref{model}, it is straightforward that
$$\mathbb{V}ar(\boldsymbol{Y}_i|W_i, \mathbf{X}_i) = \sigma^2_{\gamma}\mathbf{J}_{m_i} + \sigma^2_{\epsilon}\mathbf{I}_{m_i}$$
where $\mathbf{J}_{m_i} = \mathbf{1}_{m_i}\mathbf{1}^{\top}_{m_i}$ for all $i=1,\ldots,I$. A typical reparameterization of $\mathbb{V}ar(\boldsymbol{Y}_i|W_i, \mathbf{X}_i)$ is
$$\mathbb{V}ar(\boldsymbol{Y}_i|W_i, \mathbf{X}_i) = \sigma^2_{y|x}\mathbf{R}_i(\rho)$$
where $\sigma^2_{y|x} = \sigma^2_{\gamma} + \sigma^2_{\epsilon}$ denotes the conditional variance of the individual level outcome and
\begin{equation}\label{Ri}
\mathbf{R}_i(\rho) = (1 - \rho)\mathbf{I}_{m_i} + \rho\mathbf{J}_{m_i}
\end{equation}
is the conditional correlation matrix function of $\boldsymbol{Y}_i$ with the ICC $\rho = \sigma^2_{\gamma} / \sigma^2_{y|x}$ as the single argument. Since
$$\mathbf{R}_i(\rho)^{-1} = \frac{1}{1 - \rho}\left(\mathbf{I}_{m_i} - \frac{\rho}{1 - \rho + m_i\rho}\mathbf{J}_{m_i}\right)$$
\cite[see Eq.~(A.1), Web Appendix][]{li2018sample} for example, the function
$$d_i(\rho) = \frac{1}{1 - \rho + m_i\rho}$$
is an eigenvalue of $\mathbf{R}_i(\rho)^{-1}$ with the corresponding eigenvector is $\mathbf{1}_{m_i}$ for all $i=1,\ldots,I$. Let
\begin{equation}\label{Key}
\overline{W}_{\boldsymbol{m}}(\rho) = \frac{\displaystyle\sum^n_{i=1}m_id_i(\rho)W_i}{\displaystyle\sum^n_{i=1}m_id_i(\rho)}.
\end{equation}
Hence reparameterization of \eqref{model} yields
\begin{equation}\label{Repar}
\boldsymbol{Y}_i = \nu_1\mathbf{1}_{m_i} + \nu_2(W_i - \overline{W}_{\boldsymbol{m}}(\rho))\mathbf{1}_{m_i} + \mathbf{X}_i\boldsymbol{\nu}_3 + (W_i - \overline{W}_{\boldsymbol{m}}(\rho))\mathbf{X}_i\boldsymbol{\nu}_4 + \gamma_i\mathbf{1}_{m_i} + \boldsymbol{\epsilon}_i,
\end{equation}
where $\nu_1 = \beta_1 + \beta_2\overline{W}_{\boldsymbol{m}}(\rho)$, $\nu_2 = \beta_2$, $\boldsymbol{\nu}_3 = \boldsymbol{\beta}_3 + \boldsymbol{\beta}_4\overline{W}_{\boldsymbol{m}}(\rho)$, and $\boldsymbol{\nu}_4 = \boldsymbol{\beta}_4$. Correspondingly, the design matrix for LMM \eqref{Repar} is given by
\begin{align*}
\widetilde{\mathbf{Z}}_i(\rho) &= \begin{bmatrix}\mathbf{1}_{m_i} & (W_i - \overline{W}_{\boldsymbol{m}}(\rho))\mathbf{1}_{m_i} & \mathbf{X}_i & (W_i - \overline{W}_{\boldsymbol{m}}(\rho))\mathbf{X}_i\end{bmatrix} \\
&= \begin{bmatrix} \begin{bmatrix} 1 & W_i - \overline{W}_{\boldsymbol{m}}(\rho) \end{bmatrix}\otimes \mathbf{1}_{m_i} & \begin{bmatrix} 1 & W_i - \overline{W}_{\boldsymbol{m}}(\rho) \end{bmatrix}\otimes  \mathbf{X}_i\end{bmatrix}
\end{align*}
where $\otimes$ denotes the Kronecker product. Thus
\begin{align*}
\mathbf{1}^{\top}_{m_i}\widetilde{\mathbf{Z}}_i(\rho) &= m_i\begin{bmatrix}1 & W_i - \overline{W}_{\boldsymbol{m}}(\rho) & \boldsymbol{\theta}^{\top} & (W_i - \overline{W}_{\boldsymbol{m}}(\rho))\boldsymbol{\theta}^{\top}\end{bmatrix} \\
&= m_i\begin{bmatrix} \begin{bmatrix} 1 & W_i - \overline{W}_{\boldsymbol{m}}(\rho) \end{bmatrix}\otimes 1 & \begin{bmatrix} 1 & W_i - \overline{W}_{\boldsymbol{m}}(\rho) \end{bmatrix}\otimes  \boldsymbol{\theta}^{\top}\end{bmatrix}.
\end{align*}
Since 
$$\mathbb{V}ar(\hat{\boldsymbol{\beta}}_4|\{W_i, \mathbf{X}_i\}^I_{i=1}) = \mathbb{V}ar(\hat{\boldsymbol{\nu}}_4|\{\widetilde{\mathbf{Z}}_i(\rho)\}^I_{i=1})$$
is the lower-right $p\times p$ block of $\left(\sum^I_{i=1}\widetilde{\mathbf{Z}}_i(\rho)^{\top}\mathbf{R}_i(\rho)^{-1}\widetilde{\mathbf{Z}}_i(\rho)\right)^{-1}$, we denote
$$\sum^I_{i=1}\widetilde{\mathbf{Z}}_i(\rho)^{\top}\mathbf{R}_i(\rho)^{-1}\widetilde{\mathbf{Z}}_i(\rho) = \frac{1}{1 - \rho}\mathbf{S}_{\boldsymbol{m}}(\rho) - \frac{\rho}{1 - \rho}\mathbf{T}_{\boldsymbol{m}}(\rho),
$$
where
\begin{align*}
\mathbf{S}_{\boldsymbol{m}}(\rho) &= \sum^I_{i=1}\widetilde{\mathbf{Z}}_i(\rho)^{\top}\widetilde{\mathbf{Z}}_i(\rho) \\
&= \sum^I_{i=1}\left(\begin{bmatrix}\begin{bmatrix}1 \\ W_i - \overline{W}_{\boldsymbol{m}}(\rho)\end{bmatrix}\otimes \mathbf{1}^{\top}_{m_i} \\ \begin{bmatrix}1 \\ W_i - \overline{W}_{\boldsymbol{m}}(\rho)\end{bmatrix}\otimes \mathbf{X}^{\top}_i\end{bmatrix}\right)\left(\begin{bmatrix} \begin{bmatrix} 1 & W_i - \overline{W}_{\boldsymbol{m}}(\rho) \end{bmatrix}\otimes \mathbf{1}_{m_i} & \begin{bmatrix} 1 & W_i - \overline{W}_{\boldsymbol{m}}(\rho) \end{bmatrix}\otimes  \mathbf{X}_i\end{bmatrix}\right) \\
&=\sum^I_{i=1}m_i\begin{bmatrix}
\begin{bmatrix}1 & W_i - \overline{W}_{\boldsymbol{m}}(\rho) \\ W_i - \overline{W}_{\boldsymbol{m}}(\rho) & (W_i - \overline{W}_{\boldsymbol{m}}(\rho))^2\end{bmatrix} & \begin{bmatrix}1 & W_i - \overline{W}_{\boldsymbol{m}}(\rho) \\ W_i - \overline{W}_{\boldsymbol{m}}(\rho) & (W_i - \overline{W}_{\boldsymbol{m}}(\rho))^2\end{bmatrix}\otimes \boldsymbol{\theta}^{\top} \\
\begin{bmatrix}1 & W_i - \overline{W}_{\boldsymbol{m}}(\rho) \\ W_i - \overline{W}_{\boldsymbol{m}}(\rho) & (W_i - \overline{W}_{\boldsymbol{m}}(\rho))^2\end{bmatrix}\otimes \boldsymbol{\theta} & \begin{bmatrix}1 & W_i - \overline{W}_{\boldsymbol{m}}(\rho) \\ W_i - \overline{W}_{\boldsymbol{m}}(\rho) & (W_i - \overline{W}_{\boldsymbol{m}}(\rho))^2\end{bmatrix}\otimes \text{diag}(\boldsymbol{\theta})
\end{bmatrix}
\end{align*}
and
\begin{align*}
\mathbf{T}_{\boldsymbol{m}}(\rho) &= \sum^I_{i=1}d_i(\rho)\widetilde{\mathbf{Z}}_i(\rho)^{\top}\mathbf{J}_{m_i}\widetilde{\mathbf{Z}}_i(\rho) = \sum^I_{i=1}d_i(\rho)\widetilde{\mathbf{Z}}_i(\rho)^{\top}\mathbf{1}_{m_i}\mathbf{1}^{\top}_{m_i}\widetilde{\mathbf{Z}}_i(\rho) \\
&= \sum^I_{i=1}m^2_id_i(\rho)\left(\begin{bmatrix} \begin{bmatrix}1 \\ W_i - \overline{W}_{\boldsymbol{m}}(\rho)\end{bmatrix} \\ \begin{bmatrix}1 \\ W_i - \overline{W}_{\boldsymbol{m}}(\rho)\end{bmatrix}\otimes \boldsymbol{\theta}\end{bmatrix}\right)\left(\begin{bmatrix} \begin{bmatrix} 1 & W_i - \overline{W}_{\boldsymbol{m}}(\rho) \end{bmatrix} & \begin{bmatrix} 1 & W_i - \overline{W}_{\boldsymbol{m}}(\rho) \end{bmatrix}\otimes  \boldsymbol{\theta}^{\top}\end{bmatrix}\right) \\
&= \sum^I_{i=1}m^2_id_i(\rho)\begin{bmatrix}
\begin{bmatrix}1 & W_i - \overline{W}_{\boldsymbol{m}}(\rho) \\ W_i - \overline{W}_{\boldsymbol{m}}(\rho) & (W_i - \overline{W}_{\boldsymbol{m}}(\rho))^2\end{bmatrix} & \begin{bmatrix}1 & W_i - \overline{W}_{\boldsymbol{m}}(\rho) \\ W_i - \overline{W}_{\boldsymbol{m}}(\rho) & (W_i - \overline{W}_{\boldsymbol{m}}(\rho))^2\end{bmatrix}\otimes \boldsymbol{\theta}^{\top} \\
\begin{bmatrix}1 & W_i - \overline{W}_{\boldsymbol{m}}(\rho) \\ W_i - \overline{W}_{\boldsymbol{m}}(\rho) & (W_i - \overline{W}_{\boldsymbol{m}}(\rho))^2\end{bmatrix}\otimes \boldsymbol{\theta} & \begin{bmatrix}1 & W_i - \overline{W}_{\boldsymbol{m}}(\rho) \\ W_i - \overline{W}_{\boldsymbol{m}}(\rho) & (W_i - \overline{W}_{\boldsymbol{m}}(\rho))^2\end{bmatrix}\otimes \boldsymbol{\theta}\boldsymbol{\theta}^{\top}
\end{bmatrix}.
\end{align*}
Thus,
$$\sum^I_{i=1}\widetilde{\mathbf{Z}}_i(\rho)^{\top}\mathbf{R}_i(\rho)^{-1}\widetilde{\mathbf{Z}}_i(\rho) = \begin{bmatrix}
\mathbf{A}_{\boldsymbol{m}}(\rho) & \mathbf{A}_{\boldsymbol{m}}(\rho)\otimes \boldsymbol{\theta}^{\top} \\
\mathbf{A}_{\boldsymbol{m}}(\rho)\otimes \boldsymbol{\theta} & \mathbf{B}_{\boldsymbol{m}}(\rho)
\end{bmatrix}
$$
where
\begin{align*}
\mathbf{A}_{\boldsymbol{m}}(\rho) &= \sum^I_{i=1}\left(\frac{m_i}{1 - \rho} - \frac{\rho m_i^2d_i(\rho)}{1 - \rho}\right)\begin{bmatrix}
1 & W_i - \overline{W}_{\boldsymbol{m}}(\rho) \\
W_i - \overline{W}_{\boldsymbol{m}}(\rho) & (W_i - \overline{W}_{\boldsymbol{m}}(\rho))^2
\end{bmatrix}\\
&= \sum^I_{i=1}m_id_i(\rho)\begin{bmatrix}
1 & W_i - \overline{W}_{\boldsymbol{m}}(\rho) \\
W_i - \overline{W}_{\boldsymbol{m}}(\rho) & (W_i - \overline{W}_{\boldsymbol{m}}(\rho))^2
\end{bmatrix} \\
&= \sum^I_{i=1}m_id_i(\rho)\begin{bmatrix}
1 & W_i - \overline{W}_{\boldsymbol{m}}(\rho) \\
W_i - \overline{W}_{\boldsymbol{m}}(\rho) & W_i(1 - 2\overline{W}_{\boldsymbol{m}}(\rho)) + \overline{W}_{\boldsymbol{m}}(\rho)^2
\end{bmatrix} \\
&= \begin{bmatrix}
1 & 0 \\
0 & \overline{W}_{\boldsymbol{m}}(\rho)(1 - \overline{W}_{\boldsymbol{m}}(\rho))
\end{bmatrix}\sum^I_{i=1}m_id_i(\rho)
\end{align*}
due to that $W_i\in\{0, 1\}$ for all $i=1,\ldots,I$ and the definitions of $d_i(\rho)$ and $\overline{W}_{\boldsymbol{m}}(\rho)$ and
\begin{align*}
\mathbf{B}_{\boldsymbol{m}}(\rho) &= \frac{1}{1 - \rho}\sum^I_{i=1}m_i\begin{bmatrix}
1 & (W_i - \overline{W}_{\boldsymbol{m}}(\rho)) \\
(W_i - \overline{W}_{\boldsymbol{m}}(\rho)) & (W_i - \overline{W}_{\boldsymbol{m}}(\rho))^2
\end{bmatrix}\otimes \text{diag}(\boldsymbol{\theta}) \\
&- \frac{1}{1 - \rho}\sum^I_{i=1}\rho m^2_id_i(\rho)\begin{bmatrix}
1 & (W_i - \overline{W}_{\boldsymbol{m}}(\rho)) \\
(W_i - \overline{W}_{\boldsymbol{m}}(\rho)) & (W_i - \overline{W}_{\boldsymbol{m}}(\rho))^2
\end{bmatrix}\otimes \boldsymbol{\theta}\boldsymbol{\theta}^{\top} \\
&= \mathbf{A}_{\boldsymbol{m}}(\rho)\otimes \text{diag}(\boldsymbol{\theta}) \\
&+ \frac{1}{1 - \rho}\sum^I_{i=1}\rho m^2_id_i(\rho)\begin{bmatrix}
1 & (W_i - \overline{W}_{\boldsymbol{m}}(\rho)) \\
(W_i - \overline{W}_{\boldsymbol{m}}(\rho)) & (W_i - \overline{W}_{\boldsymbol{m}}(\rho))^2
\end{bmatrix}\otimes \left(\text{diag}(\boldsymbol{\theta}) - \boldsymbol{\theta}\boldsymbol{\theta}^{\top}\right).
\end{align*}
Since the lower-right $2p\times 2p$ block of $\left(\sum^I_{i=1}\widetilde{\mathbf{Z}}_i(\rho)^{\top}\mathbf{R}_i(\rho)^{-1}\widetilde{\mathbf{Z}}_i(\rho)\right)^{-1}$ is given by
$$\left[\mathbf{B}_{\boldsymbol{m}}(\rho) - \left(\mathbf{A}_{\boldsymbol{m}}(\rho)\otimes\boldsymbol{\theta}\right)\mathbf{A}_{\boldsymbol{m}}(\rho)^{-1}\left(\mathbf{A}_{\boldsymbol{m}}(\rho)\otimes\boldsymbol{\theta}^{\top}\right)\right]^{-1} = \left(\mathbf{B}_{\boldsymbol{m}}(\rho) - \mathbf{A}_{\boldsymbol{m}}(\rho)\otimes\boldsymbol{\theta}\boldsymbol{\theta}^{\top}\right)^{-1},$$
we consider
\begin{align*}
\mathbf{B}_{\boldsymbol{m}}(\rho) - \mathbf{A}_{\boldsymbol{m}}(\rho)\otimes\boldsymbol{\theta}\boldsymbol{\theta}^{\top} &= \left[\mathbf{A}_{\boldsymbol{m}}(\rho) + \frac{1}{1 - \rho}\sum^I_{i=1}\rho m^2_id_i(\rho)\begin{bmatrix}
1 & (W_i - \overline{W}_{\boldsymbol{m}}(\rho)) \\
(W_i - \overline{W}_{\boldsymbol{m}}(\rho)) & (W_i - \overline{W}_{\boldsymbol{m}}(\rho))^2
\end{bmatrix}\right] \\
&\otimes \left(\text{diag}(\boldsymbol{\theta}) - \boldsymbol{\theta}\boldsymbol{\theta}^{\top}\right) \\
&= \frac{1}{1 - \rho}\sum^I_{i=1}m_i\begin{bmatrix}
1 & (W_i - \overline{W}_{\boldsymbol{m}}(\rho)) \\
(W_i - \overline{W}_{\boldsymbol{m}}(\rho)) & (W_i - \overline{W}_{\boldsymbol{m}}(\rho))^2
\end{bmatrix}\otimes \left(\text{diag}(\boldsymbol{\theta}) - \boldsymbol{\theta}\boldsymbol{\theta}^{\top}\right).
\end{align*}
Therefore,
$$(\mathbf{B}_{\boldsymbol{m}}(\rho) - \mathbf{A}_{\boldsymbol{m}}(\rho)\otimes\boldsymbol{\theta}\boldsymbol{\theta}^{\top})^{-1} = \frac{1 - \rho}{\kappa(\rho)}\sum^I_{i=1}m_i\begin{bmatrix}
(W_i - \overline{W}_{\boldsymbol{m}}(\rho))^2 & -(W_i - \overline{W}_{\boldsymbol{m}}(\rho)) \\
-(W_i - \overline{W}_{\boldsymbol{m}}(\rho)) & 1
\end{bmatrix}\otimes \left(\text{diag}(\boldsymbol{\theta}) - \boldsymbol{\theta}\boldsymbol{\theta}^{\top}\right)^{-1},$$
where
$$\kappa(\rho) = \left(\sum^I_{i=1}m_i\right)\left(\sum^I_{i=1}m_i(W_i - \overline{W}_{\boldsymbol{m}}(\rho))^2\right) - \left(\sum^I_{i=1}m_i(W_i - \overline{W}_{\boldsymbol{m}}(\rho))\right)^2.$$
Note that
$$\sum^I_{i=1}m_i(W_i - \overline{W}_{\boldsymbol{m}}(\rho))^2 = \sum^I_{i=1}m_i(W_i - \overline{W}_{\boldsymbol{m}})^2 + I\bar{m}(\overline{W}_{\boldsymbol{m}}(\rho) - \overline{W}_{\boldsymbol{m}})^2$$
and
$$\sum^I_{i=1}m_i(W_i - \overline{W}_{\boldsymbol{m}}(\rho)) = I\bar{m}(\overline{W}_{\boldsymbol{m}} - \overline{W}_{\boldsymbol{m}}(\rho)).$$
We obtain
$$\kappa(\rho) \equiv I\bar{m}\sum^I_{i=1}m_i(W_i - \overline{W}_{\boldsymbol{m}})^2 = I\bar{m}\overline{W}_{\boldsymbol{m}}(1 - \overline{W}_{\boldsymbol{m}}),$$
which leads to \eqref{main result}.
\end{proof}
\end{document}